\documentclass[nohyperref]{article} \usepackage[icml]{optional}

\usepackage{booktabs} %
\usepackage{xfrac} %
\usepackage{makecell} %

\newcommand{\Bernoulli}{\mathrm{Bernoulli}}

\newcommand{\InvGamma}{\mathrm{InvGamma}}
\newcommand{\lp}[1][p]{L^{#1}} %

\providecommand{\Diag}{\mathop\mathrm{Diag}}

\newcommand{\defeq}{\triangleq}
\newcommand\mydots{\makebox[0.8em][c]{.\hfil.\hfil.}}
\usepackage{xspace}
\newcommand{\sss}{$\mathrm{S}^3$\xspace}

\newcommand{\EJ}{\mathbb{E}_J}
\newcommand{\VJ}{{\rm Var}_J}
\newcommand{\CJ}{{\rm Cov}_J}
\newcommand{\cJ}{{\rm corr}_J}
\newcommand{\minus}{\scalebox{0.5}[1.0]{$-$}}

\usepackage{amsmath}
\usepackage{amssymb}
\usepackage{mathtools}
\usepackage{amsthm}
\usepackage{bm}

\theoremstyle{plain}
\newtheorem{theorem}{Theorem}[section]
\newtheorem{proposition}[theorem]{Proposition}

\theoremstyle{definition}

\theoremstyle{remark}

\opt{icml}{
\usepackage{microtype}
\usepackage{graphicx}
\usepackage{subfigure}

\usepackage{hyperref}

\usepackage[accepted]{icml2022_template/icml2022}
\icmltitlerunning{Scalable Spike-and-Slab}
}

\opt{arxiv}{
\usepackage{lmodern}
\usepackage{geometry}
\usepackage[utf8]{inputenc} %
\usepackage[T1]{fontenc}    %
\usepackage{textcomp}
\usepackage{amsfonts}       %
\usepackage{microtype}      %

\usepackage{natbib}
\bibliographystyle{abbrvnat}

\usepackage{url} 
\usepackage[unicode=true,bookmarks=true,bookmarksnumbered=false,bookmarksopen=false,breaklinks=false,pdfborder={0 0 0},pdfborderstyle={},backref=page,colorlinks=true]{hyperref}
\hypersetup{pdftitle={Scalable Spike-and-Slab},
pdfauthor={Biswas Meng Mackey},
linkcolor=RoyalBlue,citecolor=RoyalBlue,urlcolor=RoyalBlue}
\usepackage[dvipsnames,svgnames,x11names,hyperref]{xcolor}
\usepackage{tikz}

\usepackage[ruled]{algorithm}
\usepackage{algorithmic}
}

\usepackage[capitalize,noabbrev]{cleveref}
\usepackage[textsize=tiny]{todonotes}

\begin{document}

\opt{arxiv}{
\title{Scalable Spike-and-Slab}
  \author{Niloy Biswas 
  \thanks{niloy\_biswas@g.harvard.edu} 
  \\
    Harvard University \\
    \and
    Lester Mackey \thanks{lmackey@microsoft.com} \\
    Microsoft Research New England \\
    \and
    Xiao-Li Meng \thanks{meng@stat.harvard.edu} \\
    Harvard University
    }
  \maketitle
}

\opt{icml}{
\twocolumn[
\icmltitle{Scalable Spike-and-Slab 
}

\icmlsetsymbol{equal}{*}

\begin{icmlauthorlist}
\icmlauthor{Niloy Biswas}{h}
\icmlauthor{Lester Mackey}{msr}
\icmlauthor{Xiao-Li Meng}{h}
\end{icmlauthorlist}

\icmlaffiliation{h}{Department of Statistics, Harvard University}
\icmlaffiliation{msr}{Microsoft Research New England}

\icmlcorrespondingauthor{Niloy Biswas}{niloy\_biswas@g.harvard.edu}
\icmlkeywords{Machine Learning, ICML}

\vskip 0.3in
]

\printAffiliationsAndNotice{}  %
}

\begin{abstract}
Spike-and-slab priors are commonly used for Bayesian variable selection, due to 
their interpretability and favorable statistical properties. 
However, existing samplers for spike-and-slab posteriors incur prohibitive computational costs when the number of variables is large.
In this article, we propose \emph{Scalable Spike-and-Slab} (\sss), a 
scalable Gibbs sampling implementation for high-dimensional 
Bayesian regression with the continuous spike-and-slab 
prior of \citet{george1993variableJASA}. 
For a dataset with $n$ observations and $p$ covariates, \sss 
has order $\max \{ n^2 p_t, np \}$ computational cost at iteration $t$ where $p_t$ 
never exceeds the number of covariates switching 
spike-and-slab states between iterations $t$ and $t-1$ 
of the Markov chain.
This improves upon the order  $n^2p$ per-iteration cost of state-of-the-art 
implementations as, typically, $p_t$ is substantially smaller than $p$.
We apply \sss on synthetic and real-world datasets, demonstrating orders of magnitude speed-ups over existing exact samplers and significant gains in inferential quality over approximate samplers with comparable cost.
\end{abstract}

\section{Introduction} \label{section:intro}

\subsection{Bayesian computation in high dimensions}
\label{subsection:bayes_comp}
We consider linear,  logistic, and probit regression in high dimensions, 
where the number of observations $n$ is 
smaller than the number of covariates $p$. 
This setting is common in modern applications such as 
genome-wide association studies \citep{guan2011bayesianAOAS, zhou2013polygenicPLOS}
and astronomy \citep{kelly2007someASTRO,sereno2015bayesianMONTHLY}.
In the non-Bayesian paradigm, sparse point estimates 
such as the LASSO \citep{tibshirani1994regression}, Elastic Net \citep{zou2005regularization} 
and SLOPE \citep{Bogdan2015slopeAOAS} 
offer a route to variable selection. These estimates are based on 
optimization based approaches, which are 
computationally efficient and scale to 
datasets with hundreds of thousands of covariates.  

In the Bayesian paradigm, which will be our focus, one places a prior 
on the unknown parameters of interest and considers the corresponding posterior distribution. 
Sampling algorithms such as Markov chain Monte Carlo (MCMC) are then 
used to simulate from the posterior distribution. In modern high dimensional settings, 
general-purpose MCMC algorithms can have high computational cost per iteration. 
This has kindled a line of work on tailored algorithms for Bayesian regression 
\citep[e.g.,][]{polson2013bayesianJASA, yang2016onAOS, narisetty2019skinnyJASA, johndrow2020scalableJMLR, biswas2022coupled}.
Our manuscript participates in this wider effort to scale Bayesian inference to large data
applications. Specifically, we propose computationally efficient MCMC algorithms for 
high-dimensional Bayesian linear, logistic and probit regression with spike-and-slab priors. 

\subsection{Variable selection with spike-and-slab priors}
\label{subsection:bayes_model_select}
Consider Gaussian linear regression, logistic regression, and probit regression 
with $n$ observations and $p$ covariates.  The respective likelihoods are given by
$L_{\textup{lin}}(\bm{\beta}, \sigma^2 ; \bm{X},\bm{y}) = \frac{1}{(2 \pi \sigma^2)^{n/2}} \exp \big(- \frac{\sum_{i=1}^n (y_i - \bm{x_i}^\top \bm{\beta})^2}{2\sigma^2}\big)$, 
$L_{\textup{log}}(\bm{\beta} ; \bm{X},\bm{y}) = \prod_{i=1}^n \frac{\exp(y_i \bm{x_i}^\top \bm{\beta})}{1+\exp(\bm{x_i}^\top \bm{\beta})}$, and 
$L_{\textup{prob}}(\bm{\beta} ; \bm{X},\bm{y}) = \prod_{i=1}^n \Phi( \bm{x_i}^\top \bm{\beta} )^{y_i} (1-\Phi( \bm{x_i}^\top \bm{\beta} ))^{1-y_i}$.
Here  
$\bm{X} \in \mathbb{R}^{n \times p}$ is the design matrix with rows $\bm{x_i}^\top$, 
$\bm{y} \in \mathbb{R}^n$ (for linear regression) or $\bm{y} \in \{0,1\}^n$ (for logistic and probit regression) 
is the response vector, $\bm{\beta} \in \mathbb{R}^p$ is the unknown signal, 
$\sigma^2 \in (0, \infty)$ is the unknown Gaussian noise variance, and 
$\Phi$ is the cumulative density function of ${\cal N}(0,1)$.

We focus on the high-dimensional
setting with $n \ll p$, where $\bm{\beta} \in \mathbb{R}^p$
is assumed to be sparse.  We use a continuous spike-and-slab 
prior on $\bm{\beta}$ to capture sparsity: 
\opt{arxiv}{
\begin{flalign}
& \sigma^{2} \sim \InvGamma \Big(\frac{a_0}{2}, \frac{b_0}{2} \Big), 
\hspace{16pt} \text{(for linear regression)} \nonumber \\
& \sigma^{2} = 1, \hspace{100pt} \text{(for logistic and probit regression)} \nonumber \\
& z_j \underset{j=1,\mydots,p}{\overset{i.i.d.}{\sim}} \Bernoulli(q), \nonumber \\
& \beta_j | z_j, \sigma^2 \underset{j=1,\mydots,p}{\overset{ ind }{\sim}} (1-z_j) \mathcal{N}(0, \sigma^2 \tau_0^2) + z_j \mathcal{N}(0, \sigma^2 \tau_1^2), \label{eq:cont_spike_slab}
\end{flalign}
}
\opt{icml}{
\begin{flalign}
& \sigma^{2} \sim \InvGamma \Big(\frac{a_0}{2}, \frac{b_0}{2} \Big), 
\hspace{31pt} \text{(for linear regression)} \nonumber \\
& \sigma^{2} = 1, \hspace{65pt} \text{(for logistic and probit regression)} \nonumber \\
& z_j \underset{j=1,\mydots,p}{\overset{i.i.d.}{\sim}} \Bernoulli(q), \nonumber \\
& \beta_j | z_j,  \! \sigma^2 \! \underset{j=1,\mydots,p}{\overset{ind}{\sim}} \! (1 \! - \! z_j)\mathcal{N}(0, \!  \sigma^2 \tau_0^2) \! + \! z_j\mathcal{N}(0,  \! \sigma^2 \tau_1^2), \label{eq:cont_spike_slab}
\end{flalign}
}
where 
$q\in (0,1)$, $\tau_1^2 \gg \tau_0^2 >0$, and $a_0, b_0 > 0$ are 
hyperparameters. Here, $\mathcal{N}(0, \sigma^2 \tau_0^2)$
and $\mathcal{N}(0, \sigma^2 \tau_1^2)$ correspond to the \textit{spike} and \textit{slab} parts
of the prior respectively. 
In the high-dimensional setting, 
a small constant $q \ll 1$ is often chosen, but the algorithms in this manuscript also
readily extend to hierarchical variants of \eqref{eq:cont_spike_slab} with a hyperprior placed on $q$ \citep{scott2010bayesAOS,castillo2012needlesAOS}. 

Catalyzed by the works of \citet{george1993variableJASA, george1997approachesSS}, 
continuous spike-and-slab priors are now a mainstay of Bayesian variable
selection (see the recent reviews of \citet[][Section I]{vannucci2021handbook}
and \citet{banerjee2021bayesian}).
The posterior probabilities $\mathbb{P}(z_j=1| \bm{y})$ provide a natural
interpretable approach to variable selection. The \textit{median probability model} selects all covariates $j$ such that $\mathbb{P}(z_j=1| \bm{y}) > \sfrac{1}{2}$, and it is easily fitted using Monte Carlo samples and provides the optimal predictive model in the case with orthogonal design matrix, as well as extensions with certain correlated matrices \citep{barbieri2004optimalAOS, barbieri2021medianBA}.
\citet{narisetty2014bayesianAOS} have further fine-tuned the 
optimal scaling of $q, \tau_0^2$, and $\tau_1^2$ with respect to the number of 
covariates and sample size to establish model selection consistency for 
linear regression with general design matrices in high dimensions. 

One could alternatively consider \textit{point-mass} spike-and-slab priors 
\citep[e.g.,][]{mitchell1988bayesianJASA, johnson2012bayesianJASA}, where 
$\beta_j | z_j \sim (1-z_j) \delta_0(\cdot) + z_j \mathcal{N}(0, \sigma^2 \tau_1^2)$ 
such that a degenerate Dirac distribution about zero is chosen for the spike part. 
Point-mass priors have favorable statistical properties
\citep[e.g.,][]{johnstone2004needlesAOS, castillo2012needlesAOS} 
and we hope to extend our algorithms to point-mass priors in follow-up work.

\subsection{Our contributions}
Our contributions are summarized below.
Throughout, we use $\mathcal{O}$ and $\Omega$ to respectively denote asymptotic upper and lower
bounds on computational complexity growth rates. 

Section \ref{section:mcmc} introduces 
\textit{Scalable Spike-and-Slab} (\sss), 
a computationally efficient implementation of Gibbs samplers 
for linear and logistic regression with the
prior given in \eqref{eq:cont_spike_slab}. 
Section \ref{subsection:bottlenecks} investigates the 
computational bottlenecks of state-of-the-art (SOTA) 
implementations, which require $\Omega(n^2p)$
computational cost per iteration for datasets with $n$ observations and $p$ covariates. 
Section \ref{subsection:S3} develops \sss, 
which overcomes existing computational bottlenecks by employing 
a pre-computation based strategy and requires 
$\mathcal{O}(\max\{n^2 p_t, np\})$ computational cost at iteration $t$, 
where $p_t$ is no greater than the number of covariates switching 
spike-and-slab states between iterations $t$ and $t-1$ 
of the Markov chain. 
Section \ref{subsection:complexity_analysis} analyzes 
the favourable computational complexity of \sss, showing 
that $p_t$ is typically much smaller than $p$ and that it can remain constant and even approach zero
under various limiting regimes as $p$ increases.

Section \ref{section:comparison} compares
\sss with the SOTA exact MCMC sampler and a recently 
proposed approximate MCMC sampler, which does not converge to the posterior distribution of interest.
We demonstrate that \sss offers substantially
faster numerical runtimes compared to the SOTA exact MCMC sampler, 
reporting $50\times$ speedups on synthetic datasets. 
In the same experiment, 
\sss and the approximate sampler have comparable runtimes, but the asymptotically exact \sss procedure provides more accurate variable selection.

Section \ref{section:applications} demonstrates 
the benefits of \sss
on a diverse suite of datasets, including two synthetic datasets and eight real-world experiments. 
For example, on a genome-wide association study (GWAS) dataset with 
with $n\!\approx\!2000$ and $p\!\approx\!100000$, we again observe 
$50\times$ computational speedups over the SOTA exact MCMC sampler. 
Finally, Section \ref{section:discussion} discusses directions for future work. 
The open-source packages $\href{https://CRAN.R-project.org/package=ScaleSpikeSlab}{\mathrm{ScaleSpikeSlab}}$ in 
$\mathrm{R}$ and $\mathrm{Python}$ 
(\url{www.github.com/niloyb/ScaleSpikeSlab})
implement our methods and recreate the experiments in this paper.

\section{Scalable Spike-and-Slab} \label{section:mcmc}

\subsection{Status quo and computational bottlenecks}
\label{subsection:bottlenecks}
Gibbs samplers have long been employed to sample from the 
posterior distributions corresponding to the prior in \eqref{eq:cont_spike_slab} 
\citep[e.g.,][]{george1993variableJASA, obrien2004bayesianBIOMETRICS, 
held2006bayesianBA, polson2013bayesianJASA}. The 
computational bottleneck of existing Gibbs samplers
is linked to  sampling from the full conditional of 
$\bm{\beta} \in \mathbb{R}^p$. This is given by 
\opt{arxiv}{
\begin{equation} \label{eq:full_conditional_beta}
    \bm{\beta}_{t+1} | \bm{z}_t, \sigma^2_t \sim 
    \mathcal{N} \big(\bm{\Sigma}^{-1}_t \bm{X}^\top \bm{y}, \sigma_t^2 \bm{\Sigma}^{-1}_t \big)\quad {\rm for}\quad  \bm{\Sigma}_t = \bm{X}^\top \bm{X} + \bm{D}_t,
\end{equation}
}
\opt{icml}{
\begin{equation} \label{eq:full_conditional_beta}
    \bm{\beta}_{t+1} | \bm{z}_t, \sigma^2_t \! \sim \!
    \mathcal{N} \big(\bm{\Sigma}^{-1}_t \bm{X}^\top \bm{y}, \sigma_t^2 \bm{\Sigma}^{-1}_t \big)
\end{equation}
}
\opt{icml}{for $\bm{\Sigma}_t \!=\! \bm{X}^\top \bm{X} + \bm{D}_t$},
where $t$ indexes the iteration of the Markov chain, 
and $\bm{D}_t$ is the diagonal matrix with the vector 
$\bm{z}_t \tau^{-2}_1 + (\bm{\mathrm{1}}_p-\bm{z}_t) \tau^{-2}_0$ 
populating its diagonal elements.

Sampling from \eqref{eq:full_conditional_beta} using standard matrix multiplication and a generic Cholesky 
decomposition that ignores the specific structure 
of $\bm{\Sigma}_t$ requires $\Omega(p^3)$ computational cost, 
which quickly becomes prohibitive for large $p$. 
Hereafter we will refer to this generic method as the Na\"ive Sampler. 
\citet{ishwaran2005spikeAOS} recommend separating the components 
of $\bm{\beta}$ into $B$ blocks of size $\sfrac{p}{B}$ each, and then updating each block using Gibbs
sampling, which gives a reduced $\Omega(p^3 B^{-2})$ computational cost. 
However, this cost remains prohibitive for large $p$, and using a 
larger number of blocks $B$ induces higher auto-correlation 
between successive iterations of the Gibbs sampler.
Recently, \citet{bhattacharya2016fastBIOMETRIKA} developed 
an algorithm based on the Woodbury matrix identity \citep{hager1989theSIAM}
to sample from multivariate Gaussian distributions of the 
form in \eqref{eq:full_conditional_beta}, 
which requires a more favourable $\Omega(n^2p)$ 
computational cost and is given in Algorithm \ref{algo:fast_mvn_bhattacharya}. 
\begin{algorithm}[tb]
\caption{An $\Omega(n^2p)$ sampler of \eqref{eq:full_conditional_beta}
\citep{bhattacharya2016fastBIOMETRIKA}}
   \label{algo:fast_mvn_bhattacharya}
\begin{algorithmic}
  \STATE Sample $\bm{r} \sim \mathcal{N}(0, \bm{I}_p)$, $\bm{\xi} \sim \mathcal{N}(0,\bm{I}_n)$.
  \STATE Set $\bm{u} = \bm{D}_t^{-\frac{1}{2}}\bm{r}$ and calculate $\bm{v} = \bm{X}\bm{u} + \bm{\xi}$.
  \STATE Set $\bm{v}^* = \bm{M}_t^{-1} (\frac{1}{\sigma_t} \bm{y} - \bm{v})$ for $\bm{M}_t = \bm{I}_n + \bm{X} \bm{D}_t^{-1} \bm{X}^\top$.
  \STATE {\bfseries Return} $\bm{\beta} = \sigma_t ( \bm{u} +  \bm{D}_t^{-1} \bm{X}^\top \bm{v}^* )$.
\end{algorithmic}
\end{algorithm}

For large-scale datasets with $n$ in the thousands and $p$ in the hundreds 
of thousands, as found in modern scientific applications, 
the $\Omega(n^2p)$ cost per iteration is still too high. 
This has spurred recent work on approximate 
MCMC \citep{narisetty2019skinnyJASA} for the continuous
spike-and-slab prior on logistic regression and on variational inference methods 
for the point-mass spike-and-slab prior 
\citep{titsias2011spikeNEURIPS, ray2020spikeNEURIPS, ray2021variationalJASA}.
Such approximate samplers can provide improved computational 
speeds but do not converge to the posterior distribution of interest.  

\subsection{A scalable Gibbs sampler} \label{subsection:S3}
We now develop \sss. Our key insight is that successive pre-computation can be used to reduce the computational 
cost of Algorithm \ref{algo:fast_mvn_bhattacharya}. In Algorithm \ref{algo:fast_mvn_bhattacharya}, 
the $\Omega(n^2p)$ computational cost per iteration 
arises from the calculation of the matrix $\bm{M}_t$. 
Current algorithms calculate $\bm{M}_t = \bm{I}_n + \bm{X} \bm{D}_t^{-1} \bm{X}^\top$ from scratch 
every iteration at $\Omega(n^2p)$ cost under standard matrix multiplication and then 
solve an $n$ by $n$ linear system to obtain $v^*$
at $\Omega(n^3)$ cost. 
We propose instead to use the previous state $\bm{z}_{t-1}$ and the pre-computed matrices $\bm{M}_{t-1}$ and $\bm{M}_{t-1}^{-1}$
at each step to aid the calculation of $\bm{M}_t^{-1}$. Our strategy is given below. 

Specifically, denote $\bm{X} = [\bm{X}_1,\mydots,\bm{X}_p] \in \mathbb{R}^{n \times p}$ and its sub-matrices $\bm{X}_{A_t} \defeq [\bm{X}_j : j \in A_t] \in \mathbb{R}^{n \times \|\bm{z}_t\|}$
and $\bm{X}_{A^c_t} \defeq [\bm{X}_j : j \in A^c_t] \in \mathbb{R}^{n \times (p - \|\bm{z}_t\|)}$, where  $A_t \defeq \{ j: z_{j,t} = 1 \}$ and $A_t^c = \{ j: z_{j,t} = 0 \}$,  are the ordered index sets of covariates corresponding to slab states and to spike states respectively at iteration $t$, and $\|\cdot\|_1$ is the $\lp[1]$ norm on $\{0,1\}^p$. Also denote $\Delta_t \defeq \{ j: z_{j,t} \neq z_{j,t-1} \}$, the ordered index set of
covariates which switch spike-and-slab states between iterations $t$ and $t-1$;  $\delta_t \defeq \|\bm{z}_t-\bm{z}_{t-1}\|_1 = |\Delta_t|$, the number 
of switches; 
$\bm{D}_{\Delta_t} \defeq \Diag((\bm{D}_{t})_{j,j} : j \in \Delta_t)
\in \Diag(\mathbb{R}^{\delta_t \times \delta_t})$, 
the diagonal sub-matrix of $\bm{D}_t$ composed of the diagonal entries with 
ordered indices in $\Delta_t$; and $ \bm{C}_{\Delta_t} \defeq \bm{D}_{\Delta_t}^{-1} - \bm{D}_{\Delta_{t-1}}^{-1}$. Finally, let $\bm{\tilde{M}}_{\tau_0} \defeq \bm{I}_n + \tau_0^2 \bm{X} \bm{X}^T \in \mathbb{R}^{n \times n}$ and 
$\bm{\tilde{M}}_{\tau_1} \defeq \bm{I}_n + \tau_1^2 \bm{X} \bm{X}^T \in \mathbb{R}^{n \times n}$, which 
are fixed for all iterations. 
\begin{subequations}
Under this notation, there are three expressions for $\bm{M}_t$:
\opt{arxiv}{
\begin{flalign}
\bm{M}_t &= \bm{\tilde{M}}_{\tau_0} + (\tau_1^2-\tau_0^2) \bm{X}_{A_t} \bm{X}_{A_t}^\top \label{eq:pre_compute_correction1a} \\
&= \bm{\tilde{M}}_{\tau_1} + (\tau_0^2-\tau_1^2) \bm{X}_{A^c_t} \bm{X}_{A^c_t}^\top \label{eq:pre_compute_correction1b} \\
&= \bm{M}_{t-1} + \bm{X}_{\Delta_t}  \bm{C}_{\Delta_t} \bm{X}_{\Delta_t}^\top \label{eq:pre_compute_correction1c}.
\end{flalign}
}
\opt{icml}{
\small
\begin{flalign}
\bm{M}_t &= \bm{\tilde{M}}_{\tau_0} + (\tau_1^2-\tau_0^2) \bm{X}_{A_t} \bm{X}_{A_t}^\top \label{eq:pre_compute_correction1a} \\
&= \bm{\tilde{M}}_{\tau_1} + (\tau_0^2-\tau_1^2) \bm{X}_{A^c_t} \bm{X}_{A^c_t}^\top \label{eq:pre_compute_correction1b} \\
&= \bm{M}_{t-1} + \bm{X}_{\Delta_t}  \bm{C}_{\Delta_t} \bm{X}_{\Delta_t}^\top \label{eq:pre_compute_correction1c}.
\end{flalign}
\normalsize
}
In \eqref{eq:pre_compute_correction1a}~--~\eqref{eq:pre_compute_correction1c}, 
calculating the matrix products 
$\bm{X}_{A_t} \bm{X}_{A_t}^\top$, $\bm{X}_{A^c_t} \bm{X}_{A^c_t}^\top$, and $\bm{X}_{\Delta_t}  \bm{C}_{\Delta_t} \bm{X}_{\Delta_t}^\top$ requires 
$\mathcal{O}(n^2 \|\bm{z}_t\|_1)$, $\mathcal{O}(n^2 (p-\|\bm{z}_t\|_1))$, and
$\mathcal{O}(n^2 \delta_t)$ cost respectively. Given 
$\bm{\tilde{M}}_{\tau_0}$, $\bm{\tilde{M}}_{\tau_1}$, $\bm{M}_{t-1}$, and $\bm{z}_{t-1}$, 
we evaluate whichever matrix product in 
\eqref{eq:pre_compute_correction1a}~--~\eqref{eq:pre_compute_correction1c} 
has minimal computational cost and thereby calculate $\bm{M}_t$ at the reduced cost of 
$\mathcal{O}(n^2 p_t)$ where $p_t \defeq \min \{ \|\bm{z}_t\|_1, p-\|\bm{z}_t\|_1, \delta_t \}$.
\label{eq:pre_compute_correction1}
\end{subequations}

\begin{subequations}
To calculate $\bm{M}_t^{-1}$, we consider the cases $n \leq p_t$ and $p_t < n$ separately.
When $n \leq p_t$, we calculate $\bm{M}_t^{-1}$ by directly inverting the 
calculated matrix $\bm{M}_t$ from \eqref{eq:pre_compute_correction1}, 
which requires $\mathcal{O}(n^3)$ cost.
When $p_t < n$, we apply the Woodbury matrix identity on 
\eqref{eq:pre_compute_correction1}. This gives
\opt{arxiv}{
\begin{flalign}
\bm{M}_t^{-1} &= \bm{\tilde{M}}_{\tau_0}^{-1} -\bm{\tilde{M}}_{\tau_0}^{-1} \bm{X}_{A_t} \big( (\tau_1^2-\tau_0^2)^{-1} \bm{I}_{\|\bm{z}_t\|_1} + 
\bm{X}_{A_t}^\top \bm{\tilde{M}}_{\tau_0}^{-1} \bm{X}_{A_t} \big)^{-1} \bm{X}_{A_t}^\top \bm{\tilde{M}}_{\tau_0}^{-1} 
\label{eq:pre_compute_correction2a} \\
&= \bm{\tilde{M}}_{\tau_1}^{-1} -\bm{\tilde{M}}_{\tau_1}^{-1} \bm{X}_{A_t^c} \big( (\tau_0^2-\tau_1^2)^{-1} \bm{I}_{p-\|\bm{z}_t\|_1} + 
\bm{X}_{A^c_t}^\top \bm{\tilde{M}}_{\tau_1}^{-1} \bm{X}_{A^c_t} \big)^{-1} \bm{X}_{A^c_t}^\top \bm{\tilde{M}}_{\tau_1}^{-1} \label{eq:pre_compute_correction2b} \\
&= \bm{M}_{t-1}^{-1} - \bm{M}_{t-1}^{-1}\bm{X}_{\Delta_t} \big(  \bm{C}_{\Delta_t}^{-1} + 
\bm{X}_{\Delta_t}^\top \bm{M}_{t-1}^{-1} \bm{X}_{\Delta_t} \big)^{-1} \bm{X}_{\Delta_t}^\top \bm{M}_{t-1}^{-1}. \label{eq:pre_compute_correction2c}
\end{flalign}
}
\opt{icml}{
\small
\begin{flalign}
 \! \bm{M}_t^{\minus1} & \!= \bm{\tilde{M}}_{\tau_0}^{\minus1} \label{eq:pre_compute_correction2a} \\
&\! -\bm{\tilde{M}}_{\tau_0}^{\minus1} \! \bm{X}_{A_t} \! \big( \! \frac{1}{\tau_1^2\!-\!\tau_0^2} \! \bm{I}_{\|\bm{z}_t\|_1} \! +
\! \bm{X}_{A_t}^\top \! \bm{\tilde{M}}_{\tau_0}^{\minus1} \! \bm{X}_{A_t} \! \big)^{\minus1} \! \bm{X}_{A_t}^\top \! \bm{\tilde{M}}_{\tau_0}^{\minus1} \nonumber
\\
&= \bm{\tilde{M}}_{\tau_1}^{\minus1} \label{eq:pre_compute_correction2b} \\
&\!-\bm{\tilde{M}}_{\tau_1}^{\minus1} \! \bm{X}_{A_t^c} \! \big( \! \frac{1}{\tau_0^2\!-\!\tau_1^2} \! \bm{I}_{p-\|\bm{z}_t\|_1} \! +  \! \bm{X}_{A^c_t}^\top \! \bm{\tilde{M}}_{\tau_1}^{\minus1} \! \bm{X}_{A^c_t} \! \big)^{\minus1} \! \bm{X}_{A^c_t}^\top  \! \bm{\tilde{M}}_{\tau_1}^{\minus1} \nonumber \\
&= \bm{M}_{t\minus1}^{\minus1} \label{eq:pre_compute_correction2c} \\
&-\bm{M}_{t\minus1}^{\minus1}\bm{X}_{\Delta_t} \big(  \bm{C}_{\Delta_t}^{\minus1} + 
\bm{X}_{\Delta_t}^\top \bm{M}_{t\minus1}^{\minus1} \bm{X}_{\Delta_t} \big)^{\minus1} \bm{X}_{\Delta_t}^\top \bm{M}_{t\minus1}^{\minus1}. \nonumber
\end{flalign}
\normalsize
}
Given $\bm{\tilde{M}}_{\tau_0}^{-1}$, $\bm{\tilde{M}}_{\tau_1}^{-1}$, $\bm{M}_{t-1}^{-1}$ and $\bm{z}_{t-1}$, 
we evaluate whichever expression in \eqref{eq:pre_compute_correction2a}~--~\eqref{eq:pre_compute_correction2c}
has minimal computational cost to calculate $\bm{M}_t^{-1}$. 
Similar to \eqref{eq:pre_compute_correction1}, this requires 
$\mathcal{O}(n^2 p_t)$ computational cost,
which arises from matrix inversion and multiplication. 
\label{eq:pre_compute_correction2}
\end{subequations}

Overall, this strategy of using the previous state $\bm{z}_{t-1}$ and the 
pre-computed matrices $\bm{\tilde{M}}_{\tau_0}$, $\bm{\tilde{M}}_{\tau_0}^{-1}$, 
$\bm{\tilde{M}}_{\tau_1}$, $\bm{\tilde{M}}_{\tau_1}^{-1}$, $\bm{M}_{t-1}$ and $\bm{M}_{t-1}^{-1}$, 
reduces the computational cost of calculating the matrices $\bm{M}_{t}$ and $\bm{M}_{t}^{-1}$ from
$\Omega(n^2p)$ (as in all current implementations of Algorithm
\ref{algo:fast_mvn_bhattacharya}) to $\mathcal{O}(n^2 p_t)$.
As we show in Sections \ref{subsection:complexity_analysis} and 
\ref{section:comparison}, in many large-scale
applications $p_t$ is orders of magnitude smaller 
than both $n$ and $p$, yielding substantial improvements in 
computational efficiency. 
Furthermore, we emphasize that the matrices $\bm{\tilde{M}}_{\tau_0}$, 
$\bm{\tilde{M}}_{\tau_0}^{-1}$, $\bm{\tilde{M}}_{\tau_1}$, $\bm{\tilde{M}}_{\tau_1}^{-1}$ 
are fixed for all iterations, and the state $\bm{z}_{t-1}$ and matrices 
$\bm{M}_{t-1}$ and $\bm{M}_{t-1}^{-1}$ only need to be stored temporarily to 
generate samples for iteration $t$ and can be deleted after. 
Therefore \sss requires minimal additional memory compared 
to current implementations. 

The full Gibbs samplers for Bayesian linear, logistic, 
and probit regression which make use of this pre-computation are given in Algorithms \ref{algo:linear_spike_slab} and \ref{algo:logistic_probit_spike_slab}. 
The Gibbs samplers for logistic and probit regression are based on data augmentation strategies (see, e.g., \citet{obrien2004bayesianBIOMETRICS, narisetty2019skinnyJASA}), and the Gibbs sampler for logistic regression 
requires an adjusted pre-computation strategy with $\mathcal{O} \big(\max\{ n^2 p_t, n^3, np \} \big)$ cost.
Appendix \ref{appendices:algo_derivations} contains derivations and details of Algorithms \ref{algo:linear_spike_slab} and \ref{algo:logistic_probit_spike_slab} (the implementation of logistic regression is based on a scaled $t$-distribution 
approximation to the logistic distribution, as commonly done in the literature; see \cite{narisetty2019skinnyJASA}). 

\begin{algorithm}
\caption{Bayesian linear regression with \sss}
\label{algo:linear_spike_slab}
{\bfseries Input:} State $ \bm{C}_t \defeq (\bm{\beta}_t, \bm{z}_t, \sigma_t^2) \in 
\mathbb{R}^p \times \{0,1\}^p \times (0,\infty)$,
state $\bm{z}_{t-1}$, and matrices $\bm{M}_{t-1}, \bm{M}_{t-1}^{-1}$.
\vskip 0.1in
\begin{algorithmic}[1]
  \STATE Calculate $p_t$ and use $\eqref{eq:pre_compute_correction1}$ to calculate $\bm{M}_t$. 

  \algorithmicif{ $p_t \geq n$ } \algorithmicthen { invert $\bm{M}_t$ to calculate $\bm{M}_t^{-1}$ }
  \algorithmicelse { use \eqref{eq:pre_compute_correction2} to calculate $\bm{M}_t^{-1}$.}
  
 \STATE Sample $\bm{\beta}_{t+1} | \bm{z}_t, \sigma_t^2$ using Algorithm \ref{algo:fast_mvn_bhattacharya} from \\
 \begin{center}
     $\mathcal{N} \big( \bm{\Sigma}^{-1}_t \bm{X}^\top \bm{y},  \sigma_t^2 \bm{\Sigma}^{-1}_t \big) \text{ for } \bm{\Sigma}_t = \bm{X}^\top \bm{X} + \bm{D}_t.$
 \end{center}
  \STATE Sample each $z_{j,t+1} | \bm{\beta}_{t+1}, \sigma_t^2$ independently from
  \begin{center}
      $\Bernoulli \Big(\frac{q \mathcal{N}(\beta_{j,t+1}; 0, \sigma_t^2 \tau_1^2)}{q \mathcal{N}(\beta_{j,t+1}; 0, \sigma_t^2 \tau_1^2) + (1-q) \mathcal{N}(\beta_{j,t+1}; 0, \sigma_t^2 \tau_0^2)} \Big).$
  \end{center}
  for $j=1,\mydots,p$.
  \STATE Sample $\sigma^{2}_{t+1} | \bm{\beta}_{t+1}, \bm{z}_{t+1}$ from
  \begin{center}
      $\InvGamma \Big( \frac{a_0+n+p}{2}, \frac{b_0 + \| \bm{y} - \bm{X} \bm{\beta}_{t+1} \|_2^2 + \bm{\beta}_{t+1}^\top \bm{D}_{t+1} \bm{\beta}_{t+1} }{2} \Big).$
  \end{center}
\end{algorithmic}
\vskip 0.1in
{\bfseries Output:} $ \bm{C}_{t+1} = (\bm{\beta}_{t+1}, \bm{z}_{t+1}, \sigma^2_{t+1})$, 
$\bm{z}_t$, $\bm{M}_{t}, \bm{M}_{t}^{-1}$.
\end{algorithm}

\begin{algorithm}
\caption{Bayesian logistic \& probit regression with \sss}
\label{algo:logistic_probit_spike_slab}
{\bfseries Input:} 
State $ \bm{C}_t \defeq (\bm{\beta}_t, \bm{z}_t, \bm{\tilde{y}}_t, \bm{\tilde{\sigma}}^2_t) \in \mathbb{R}^p \times \{0,1\}^p \times \mathbb{R}^n \times (0,\infty)^n$,
states $\bm{z}_{t-1}$, $\bm{\tilde{\sigma}}^2_{t-1}$, and matrices $\bm{M}_{t-1}, \bm{M}_{t-1}^{-1}$. 
\vskip 0.1in
\begin{algorithmic}[1]
  \STATE \textit{Logistic regression:}
  Use pre-computation (see Appendix \ref{appendices:logistic}) to calculate $\bm{M}_t \defeq \bm{I}_n + \bm{W}_t^{-1/2} \bm{X} \bm{D}_t^{-1} \bm{X}^\top \bm{W}_t^{-1/2}$ for 
  $\bm{W}_t = \Diag(\bm{\tilde{\sigma}}^2_t)$. Invert $\bm{M}_t$ to calculate $\bm{M}_t^{-1}$.
  
  \textit{Probit regression:} Calculate $p_t$ and use $\eqref{eq:pre_compute_correction1}$ to calculate $\bm{M}_t \defeq \bm{I}_n + \bm{X} \bm{D}_t^{-1} \bm{X}^\top$. 
  \algorithmicif{ $p_t \geq n$ } \algorithmicthen { invert $\bm{M}_t$ to calculate $\bm{M}_t^{-1}$ }
  \algorithmicelse { use \eqref{eq:pre_compute_correction2} to calculate $\bm{M}_t^{-1}$.}
  
 \STATE Sample $\bm{\beta}_{t+1} | \bm{z}_t, \bm{\tilde{y}}_t, \bm{\tilde{\sigma}}_t^2$ using Algorithm \ref{algo:fast_mvn_bhattacharya} from \\
 \begin{center}
     $\mathcal{N} \big( \bm{\Sigma}^{-1}_t \bm{X}^\top \bm{W}^{-1}_t \bm{\tilde{y}}_t, \bm{\Sigma}^{-1}_t \big) \text{ for } \bm{\Sigma}_t\!=\!\bm{X}^\top \! \bm{W}_t^{-1} \! \bm{X}\!+\!\bm{D}_t.$
 \end{center}
  \STATE Sample each $z_{j,t+1} | \bm{\beta}_{t+1}, \bm{\tilde{y}}_t, \bm{\tilde{\sigma}}_t^2$ independently from
  \begin{center}
      $\Bernoulli \Big(\frac{q \mathcal{N}(\beta_{j,t+1}; 0, \tau_1^2)}{q \mathcal{N}(\beta_{j,t+1}; 0,  \tau_1^2) + (1-q) \mathcal{N}(\beta_{j,t+1}; 0, \tau_0^2)} \Big)$
  \end{center}
  for $j=1,\mydots,p$.
  \STATE Sample each $\tilde{y}_{i,t+1} | \bm{\beta}_{t+1}, \bm{z}_{t+1}, \bm{\tilde{\sigma}}_t^2$ independently from
  \begin{center}
$\mathcal{N}( \bm{x_i}^\top  \bm{\beta}_{t+1}, \tilde{\sigma}_{i,t}^2) \mathrm{I}_{[0, +\infty)} \quad \text{if } y_i=1,$ $\quad 
\mathcal{N}( \bm{x_i}^\top  \bm{\beta}_{t+1}, \tilde{\sigma}_{i,t}^2) \mathrm{I}_{(-\infty, 0)} \quad \text{if } y_i=0$
  \end{center}
  for $i=1,\mydots,n$.
  \STATE 
  \textit{Logistic regression:} Sample each $\tilde{\sigma}_{i,t+1}^2 | \bm{\beta}_{t+1}, \bm{z}_{t+1}, \bm{\tilde{y}}_{t+1}$ independently from
  \begin{center}
      $\InvGamma \Big( \frac{v+1}{2}, \frac{w^2 \nu + 
      (\tilde{y}_{i,t+1} - \bm{x_i}^\top \bm{\beta}_{t+1} )^2}{2} \Big)$
  \end{center}
  for $i=1,\mydots,n$, where $\nu \defeq 7.3$ and $w^2 \defeq \sfrac{\pi^2(\nu-2)}{(3 \nu)}$ are constants.
  
  \textit{Probit regression:} Set each $\tilde{\sigma}_{i,t+1}^2=1$ for $i=1,\mydots,n$.
\end{algorithmic}
\vskip 0.1in
{\bfseries Output:} $\bm{C}_{t+1} \!=\! (\bm{\beta}_{t+1}, \bm{z}_{t+1}, \bm{\tilde{y}}_{t+1}, \bm{\tilde{\sigma}}^2_{t+1})$, $\bm{z}_t$, $\bm{M}_{t}$, $\bm{M}_{t-1}^{-1}$
\end{algorithm}

\newpage

\subsection{Analysis of computational complexity}
\label{subsection:complexity_analysis}

We now investigate the favorable computational 
complexity of Algorithms \ref{algo:linear_spike_slab} and 
\ref{algo:logistic_probit_spike_slab}. Proposition \ref{prop:comp_cost},
proved in Appendix \ref{appendices:proofs}, 
gives the computational cost of these Gibbs samplers for linear and 
logistic regression, showing an improvement over 
existing implementations which have $\Omega(n^2p)$ cost. 

\begin{proposition}[Computational cost]
\label{prop:comp_cost}
Algorithm \ref{algo:linear_spike_slab} 
and Algorithm \ref{algo:logistic_probit_spike_slab} for 
probit regression both have a computational cost of 
$\mathcal{O} \big(\max\{ n^2 p_t, np \} \big)$ at iteration $t$, and 
Algorithm \ref{algo:logistic_probit_spike_slab} for  
logistic regression has a computational cost of 
$\mathcal{O} \big(\max\{ n^2 p_t, n^3, np \} \big)$
at iteration $t$, where 
$p_t = \min \{ \|\bm{z}_t\|_1, p-\|\bm{z}_t\|_1, \delta_t \}$ for 
$\delta_t=\|\bm{z}_t-\bm{z}_{t-1}\|_1$. 
\end{proposition}

In Proposition \ref{prop:comp_cost}, $\|\bm{z}_t\|_1$ and $p-\|\bm{z}_t\|_1$ are the number of 
slab covariates and the number of spike covariates respectively at iteration $t$ of the Markov chain, 
and $\delta_t$ is the number of covariates switching spike-and-slab states between iterations 
$t$ and $t-1$ of the Markov chain. Note that $p_t \leq \min \{ \|\bm{z}_t\|_1, p-\|\bm{z}_t\|_1 \} \leq p/2$ directly. In practice, there are a variety of scenarios under which $p_t$ is significantly smaller than $p/2$.

\paragraph{Sparse $\bm{z}_t$.} Whenever $\bm{z}_t$ is sparse relative to the full dimensionality $p$, 
we have $p_t \leq \|\bm{z}_t\|_1 \ll p$. Sparsity in $\bm{z}_t$ is a common occurrence 
in high-dimensional regression with a sparse signal vector $\bm{\beta}^* \in \mathbb{R}^p$, as 
$\|\bm{z}_t\|_1$ often closely approximates the true sparsity $s \defeq \|\bm{\beta}^*\|_0$ 
with high probability. This occurs, for instance, in the settings of \citet{narisetty2014bayesianAOS} 
and \citet{narisetty2019skinnyJASA}, where strong model selection consistency of the continuous 
spike-and-slab posterior is established for linear and logistic regression respectively.

\paragraph{Posterior concentration.} 
Even when $\min \{ \|\bm{z}_t\|_1, p-\|\bm{z}_t\|_1 \}$ is comparable to $p/2$,  
concentration of the spike-and-slab posterior targeted by the Gibbs sampler can lead to $\delta_t$ and hence $p_t$ remaining much smaller than $p/2$. Proposition \ref{prop:expected_comp_cost}, proved in Appendix \ref{appendices:proofs}, 
calculates the expectation of $\delta_t$ explicitly in terms of 
the posterior distribution that is targeted by our algorithms and 
the auto-correlation of the states $(\bm{z}_t)_{t \geq 0}$. 

\begin{proposition}[Expected spike-and-slab swap count]
\label{prop:expected_comp_cost}
For $\delta_t$ as given in Proposition \ref{prop:comp_cost}, 
\opt{icml}{
\begin{flalign}
\mathbb{E}[ \delta_t ] & = \sum_{j=1}^p \mathbb{P}(z_{j,t}=1)\mathbb{P}(z_{j,t-1}=0) + \nonumber \\
& \mathbb{P}(z_{j,t}=0)\mathbb{P}(z_{j,t-1}=1) -2\mathrm{cov}(z_{j,t},z_{j,t-1}). \label{eq:expected_comp_cost1}
\end{flalign}
}
\opt{arxiv}{
\begin{flalign}
\mathbb{E}[ \delta_t ] = \sum_{j=1}^p \mathbb{P}(z_{j,t}=1)\mathbb{P}(z_{j,t-1}=0) + \mathbb{P}(z_{j,t}=0)\mathbb{P}(z_{j,t-1}=1) -2\mathrm{cov}(z_{j,t},z_{j,t-1}). \label{eq:expected_comp_cost1}
\end{flalign}
}

Suppose the Markov chain 
generated by Algorithm \ref{algo:linear_spike_slab} or \ref{algo:logistic_probit_spike_slab} 
is at its stationary distribution $\pi$ at iteration $t-1$. Then,
\begin{equation} \label{eq:expected_comp_cost2}
\mathbb{E}[ \delta_t ] = 2 \sum_{j=1}^p \mathrm{var}_{\pi}(z_{j,t}) (1-\mathrm{corr}_{\pi}(z_{j,t},z_{j,t-1})).    
\end{equation}
\end{proposition}
In \eqref{eq:expected_comp_cost2}, note
$\mathrm{var}_{\pi}(z_{j,t}) = \mathbb{P}_{\pi}(z_{j,t}=0) \mathbb{P}_{\pi}(z_{j,t}=1) \leq \sfrac{1}{4}$ 
for each component $j$, with equality only when 
$z_{j,t} \sim \Bernoulli(\sfrac{1}{2})$. 
Therefore all components $j$ with $\mathbb{P}_{\pi}(z_{j,t}=1)$ close to $0$ or $1$ 
do not contribute significantly
towards $\delta_t$ in expectation. 
Such posterior concentration is guaranteed whenever $\bm{z}_t$ is convergent, be it to the true model selection vector as in \citet{narisetty2014bayesianAOS} and \citet{narisetty2019skinnyJASA} or to any other value.
In such circumstances we can have $\delta_t = \| \bm{z}_t - \bm{z}_{t-1} \|_1=o(p)$ and 
even $\delta_t = o(\|\bm{z}_t\|_1)$, regardless of the magnitude of $\|\bm{z}_t\|_1$.

\paragraph{High auto-correlation.} 
High auto-correlation of the Gibbs sampler can also lead 
to smaller values of $\delta_t$ and hence $p_t$. 
We already see from Proposition \ref{prop:expected_comp_cost} that, even for  components $j$ with bimodal 
marginal posterior distributions, high auto-correlation between 
successive states $z_{j,t-1}$ and $z_{j,t}$ can yield lower 
$\delta_t$ in expectation. 
Proposition \ref{prop:pt_empirical_bound}, proved in 
Appendix~\ref{appendices:proofs}, provides an additional exact expression for $\delta_t$ in terms of the empirical 
correlation between $\bm{z}_t$ and $\bm{z}_{t-1}$.

\begin{proposition}[Swap count decomposition]\label{prop:pt_empirical_bound}
Let $\tau_t=\sqrt{\|\bm{z}_t\|_1(p-\|\bm{z}_t\|_1)}$ and 
$\rho_t$ be the empirical correlation between $\bm{z}_t$ and $\bm{z}_{t-1}$
(that is, the correlation between $z_{J, t}$ and $z_{J, t-1}$ when $J$ is  
uniform on $\{1, \ldots, p\}$). Then, for $\delta_t$ as given in Proposition \ref{prop:comp_cost}, 
\begin{equation} \label{eq:exactp}
    \delta_t = \|\bm{z}_t\|_1+\|\bm{z}_{t-1}\|_1 - \frac{2 \|\bm{z}_t\|_1 \|\bm{z}_{t-1}\|_1 + 2 \rho_t \tau_t \tau_{t-1}}{p}.
\end{equation}
\end{proposition}
Since $|\rho_t|\le 1$, Proposition \ref{prop:pt_empirical_bound} implies
\begin{flalign}
\delta_t &\ge \frac{\big(\|\bm{z}_t\|_1 - \|\bm{z}_{t-1}\|_1\big)^2}{p} + \frac{(\tau_t-\tau_{t-1})^2}{p} \ {\rm and} \label{eq:lowerbound}\\
\delta_t &\le \frac{\big(\|\bm{z}_t\|_1 - \|\bm{z}_{t-1}\|_1\big)^2}{p} + \frac{(\tau_t+\tau_{t-1})^2}{p}. \label{eq:upperbound}
\end{flalign}
The lower bound in \eqref{eq:lowerbound} is a good approximation to $\delta_t$ when $\rho_t$ is close to one,  which is the case either when the Gibbs sampler is converging (such that $\bm{z}_t$ becomes stable) or when it gets stuck (such that $\bm{z}_t$ changes slowly with $t$). In either case, $\tau_t$ exhibits similar ``stable/stuck'' behavior, implying that the lower bound itself will be close to zero. This suggests $\delta_t$ and hence $p_t$ is close to zero when $\rho_t$ is close to 1, even if $\min\{\|\bm{z}_t\|_1, p-\|\bm{z}_t\|_1\}$ is not negligible.

Motivated by such discussions and theoretical analysis, 
we now empirically examine how $p_t$ grows as the number of 
observations $n$, the number of covariates $p$,  
and the sparsity $s$ of the true signal varies. 
Figure \ref{fig:comp_cost_sims} is based on synthetic 
linear regression datasets. For each dataset, one Markov chain is
generated using Algorithm \ref{algo:linear_spike_slab}
to target the corresponding spike-and-slab posterior, 
from which the mean and one standard error bars of $(p_t)_{t = 5000}^{10000}$ are plotted. 
For Figure \ref{fig:comp_cost_sims} (Left), we consider datasets 
with $n=100$, varying $p$ with $p \geq n$,
and a sparse true signal $\bm{\beta}^* \in \mathbb{R}^p$ with 
components $\bm{\beta}^*_j = 2 \mathrm{I}\{j \leq s\}$
for sparsity $s=10$, and noise standard deviation $\sigma^*=2$. 
For Figure \ref{fig:comp_cost_sims} (Center), we consider datasets 
with $p=1000$, varying $n$ with $n \leq p$, $s=10$, and 
$\sigma^*=2$. For Figure \ref{fig:comp_cost_sims} (Right), we consider datasets 
with $n=10s$, $p=1000$, and $\sigma^*=2$ for varying $s \geq 1$.
Details of the synthetically generated datasets are in
Appendix \ref{appendices:datasets}. 

Figure \ref{fig:comp_cost_sims} (Left) shows that 
both $p_t$ is substantially smaller than both $p$ and $n$
and that it does not increase with the number of covariates $p$.
Figure \ref{fig:comp_cost_sims} (Center) shows that 
$p_t$ tends to zero as $n$ increases. 
Figure \ref{fig:comp_cost_sims} (Right) shows that  
$p_t$ decreases 
as the sparsity $s$ increases. All figures suggest that $p_t$ is 
controlled by $\delta_t$ in these settings, because $\|\bm{z}_t\|_1$ 
takes values close to $s$, and $p-\|\bm{z}_t\|_1$ tends 
to be much larger than $\|\bm{z}_t\|_1$.
Overall, Figure \ref{fig:comp_cost_sims}
highlights that not only does $p_t$ tend to be substantially
smaller than $p$, but it also tends to
be smaller than both $n$ and $s$. 
By Proposition \ref{prop:comp_cost}, 
this showcases the substantially lower computational cost of \sss 
compared to current implementations which cost
$\Omega(n^2p)$ per iteration.

\opt{icml}{
\begin{figure}[!]
\begin{center}
\centerline{\includegraphics[width=1\columnwidth]{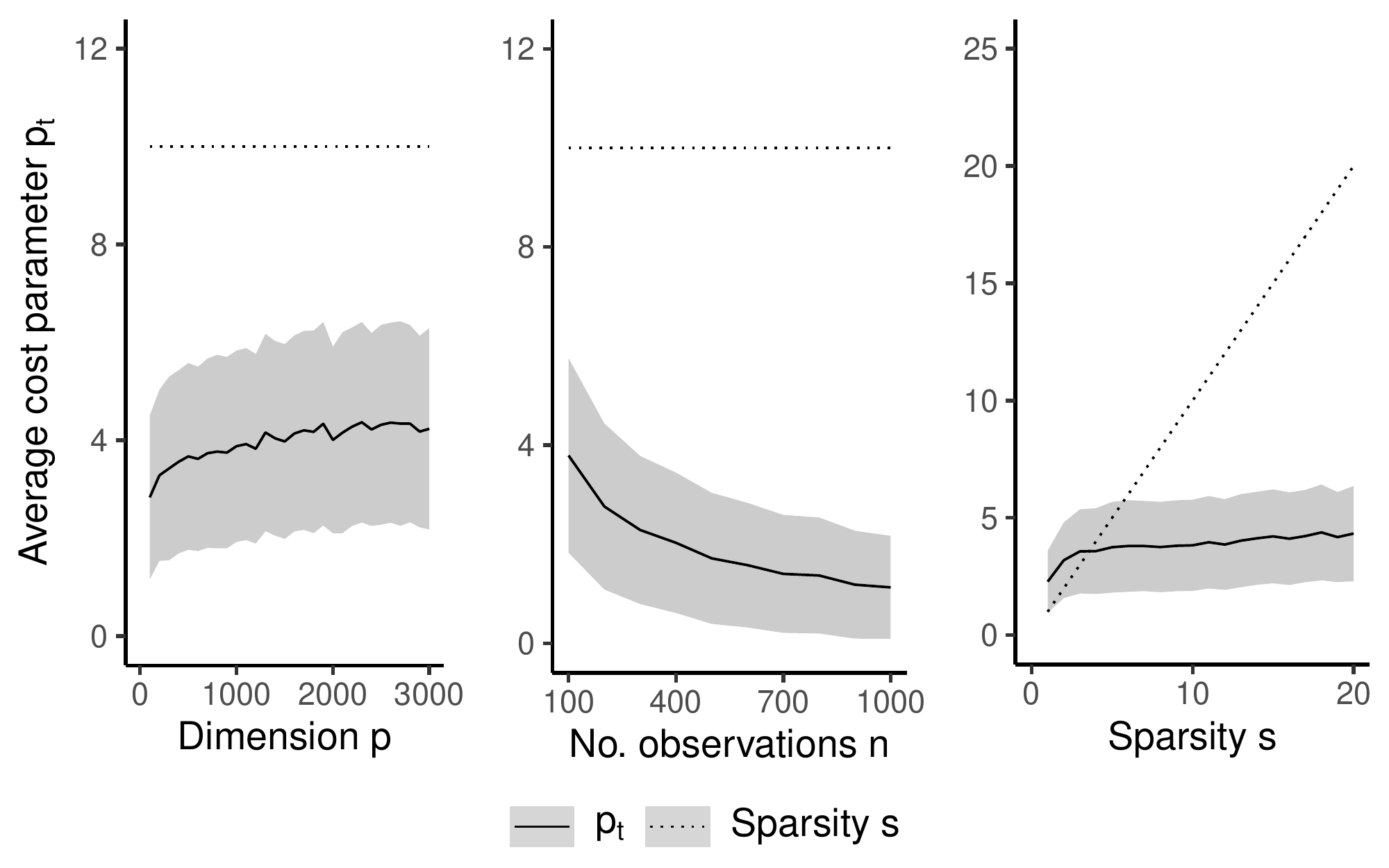}}
\vspace*{-0.05in}
\caption{The \sss cost parameter $p_t$ for averaged over iterations $5000 < t \leq 10000$
with one standard error bars, 
for synthetic linear regression datasets with varying number of 
covariates $p$ (\textbf{Left}), varying number of observations $n$ (\textbf{Center}), and 
varying sparsity $s$ (\textbf{Right}). The ground truth sparsity $s$ is also plotted for comparison.
See Section \ref{subsection:complexity_analysis} for details.}
\label{fig:comp_cost_sims}
\end{center}
\vskip -0.2in
\end{figure}
}

\opt{arxiv}{
\begin{figure}[!]
\vskip 0.2in
\begin{center}
\centerline{\includegraphics[width=1\columnwidth]{images/scaling_plot_arxiv.pdf}}
\vspace*{-0.05in}
\caption{The \sss cost parameter $p_t$ for averaged over iterations $5000 < t \leq 10000$
with one standard error bars, 
for synthetic linear regression datasets with varying number of 
covariates $p$ (Left), varying number of observations $n$ (Center), and 
varying sparsity $s$ (Right). The ground truth sparsity $s$ is also plotted for comparison.
See Section \ref{subsection:complexity_analysis} for details.}
\label{fig:comp_cost_sims}
\end{center}
\vskip -0.2in
\end{figure}
}

\subsection{Extensions to Scalable Spike-and-Slab} \label{section:S3plus}
With additional memory capacity and pre-computation, we can further improve the per-iteration costs of \sss. 

For the matrices $\bm{\tilde{M}}_{\tau_0}^{-1}$ and $\bm{\tilde{M}}_{\tau_1}^{-1}$ in Section \ref{subsection:S3}, suppose the matrices $\bm{X}^\top \bm{X}$, $\bm{X}^\top \bm{\tilde{M}}_{\tau_0}^{-1} \bm{X}$, and $\bm{X}^\top \bm{\tilde{M}}_{\tau_1}^{-1} \bm{X}$ are pre-computed. This initial step requires $\mathcal{O}(n p^2)$ computational cost and $\mathcal{O}(p^2)$ memory. Then the matrices $\bm{X}_{A_t}^\top \bm{X}_{A_t}$ and $\bm{X}_{A^c_t}^\top \bm{X}_{A^c_t}$ in \eqref{eq:pre_compute_correction1a}~--~\eqref{eq:pre_compute_correction1b} correspond to pre-computed sub-matrices of $\bm{X}^\top \bm{X}$, and
calculating $\bm{M}_t$ using \eqref{eq:pre_compute_correction1a}~--~\eqref{eq:pre_compute_correction1b}
at iteration $t$ involves matrix addition which only requires $\mathcal{O}(n^2)$ cost. Similarly, matrices $\bm{X}_{A_t}^\top \bm{\tilde{M}}_{\tau_0}^{-1} \bm{X}_{A_t}$ and $\bm{X}_{A^c_t}^\top \bm{\tilde{M}}_{\tau_1}^{-1} \bm{X}_{A^c_t}$ in \eqref{eq:pre_compute_correction2a}~--~\eqref{eq:pre_compute_correction2b} correspond to pre-computed sub-matrices of $\bm{X}^\top \bm{\tilde{M}}_{\tau_0}^{-1} \bm{X}$ and $\bm{X}^\top \bm{\tilde{M}}_{\tau_1}^{-1} \bm{X}$ respectively and do not need to be recalculated at iteration $t$. Therefore calculating $\big( (\tau_1^2-\tau_0^2)^{-1} \bm{I}_{\|\bm{z}_t\|_1} + \bm{X}_{A_t}^\top \bm{\tilde{M}}_{\tau_0}^{-1} \bm{X}_{A_t} \big)^{-1}$ or $\big( (\tau_0^2-\tau_1^2)^{-1} \bm{I}_{p-\|\bm{z}_t\|_1} + 
\bm{X}_{A^c_t}^\top \bm{\tilde{M}}_{\tau_1}^{-1} \bm{X}_{A^c_t} \big)^{-1}$ in \eqref{eq:pre_compute_correction2a}~--~\eqref{eq:pre_compute_correction2b} at each iteration $t$ only requires $\mathcal{O}(\|z_t\|_1^3)$ or $\mathcal{O}((p-\|z_t\|_1)^3)$ cost respectively. 

To sample from \eqref{eq:full_conditional_beta},  consider the cases $n \leq \min \{\|z_t\|_1, p-\|z_t\|_1\}$ and $\min \{\|z_t\|_1, p-\|z_t\|_1\} < n$ separately. When $n \leq \min \{\|z_t\|_1, p-\|z_t\|_1\}$, we calculate $\bm{M}_t^{-1}$ by directly inverting the calculated matrix $\bm{M}_t$ from \eqref{eq:pre_compute_correction1}, which requires $\mathcal{O}(n^3)$ cost.
When $\min \{\|z_t\|_1, p-\|z_t\|_1\} < n$, we avoid calculating $M_t^{-1}$ explicitly and instead calculate the matrix vector product $M_t^{-1}(\frac{1}{\sigma_t}\bm{y}-\bm{v})$ in Algorithm \ref{algo:fast_mvn_bhattacharya} right-to-left, using whichever expression in \eqref{eq:pre_compute_correction2a}~--~\eqref{eq:pre_compute_correction2b} has minimal computational cost. 
Overall, now the Gibbs samplers for linear and probit regression require only $\mathcal{O}(\max\{ \min \{\|z_t\|_1, p-\|z_t\|_1,n\}^3, np \})$ computational cost at iteration $t$. This provides lower computational cost for \sss linear and probit regression whenever $\min \{\|z_t\|_1, p-\|z_t\|_1,n\}^3 < n^2 p_t$. 
A similar extension for logistic regression requires only  $\mathcal{O}(\max\{ n^3, np \})$ computational cost at iteration $t$ and is given in Appendix \ref{appendices:logistic}.

\section{Comparison with Alternatives} \label{section:comparison}
In this section we compare \sss with the na\"{i}ve 
sampler, the SOTA exact MCMC sampler based on the sampling algorithm of
\citet{bhattacharya2016fastBIOMETRIKA}, and the Skinny Gibbs approximate MCMC sampler of \citet{narisetty2019skinnyJASA} for logistic regression. 
Table \ref{table:comparison} highlights the favorable computational cost of \sss compared to the na\"{i}ve and SOTA samplers. 
The Skinny Gibbs sampler typically has lower computational cost compared to \sss for logistic regression, and 
can have lower or higher computational cost than
\sss for probit regression depending on whether $\|\bm{z}_t\|_1^2 \le n p_t$ or not. 
However, unlike \sss, the Skinny Gibbs sampler does not converge to the correct posterior distribution.

\begin{table}[!]
\caption{\label{table:comparison}
Comparison of \sss with alternatives 
($p_t = \min \{ \|\bm{z}_t\|_1, p-\|\bm{z}_t\|_1, \delta_t \}$
for $\delta_t = \|\bm{z}_t-\bm{z}_{t-1}\|_1$). 
}
\vskip -0.2in
\begin{center}
\begin{scriptsize}
\begin{sc}
\begin{tabular}{lccr}
\toprule
MCMC Sampler & Cost & 
\makecell{Converges to \\
posterior} \\
\midrule 
Na\"{i}ve & $\Omega(p^3)$ & $\surd$ \\
State-of-the-art & $\Omega(n^2p)$ & $\surd$ \\
Skinny Gibbs & $\Omega(\max \{ n \|\bm{z}_t \|^2_1, np \})$ & $\times$ \\
$\mathbf{S^3}$ (linear and probit) & ${\mathcal{O}(\max \{ n^2 p_t, np \})}$ & $\mathbf{\surd}$ \\
$\mathbf{S^3}$ (logistic) & ${\mathcal{O}(\max \{ n^2 p_t, n^3, np \})}$ & $\mathbf{\surd}$ \\
\bottomrule
\end{tabular}
\end{sc}
\end{scriptsize}
\end{center}
\normalsize
\end{table}

To assess the practical impact of  computational cost and asymptotic bias,
Figures \ref{fig:time_comparison} and \ref{fig:stat_comparison} 
compares the numerical runtimes and statistical performance of 
\sss with the SOTA sampler and the Skinny Gibbs sampler.
For the Skinny Gibbs sampler, we use the $\mathrm{skinnybasad}$ $\mathrm{R}$ package 
of \citet{narisetty2019skinnyJASA}, which implements only logistic regression.
We consider synthetically generated datasets 
with a true signal $\bm{\beta}^* \in \mathbb{R}^p$ where 
$\beta^*_j = 2 \mathrm{I}\{j \leq s\}$ for sparsity $s$.
We consider datasets with $n=10s$ observations and 
$p=100s$ covariates for varying sparsity $s \geq 1$.
Details of the synthetically generated dataset are in Appendix
\ref{appendices:datasets}. For each synthetic dataset, 
we run \sss for logistic and probit regression and 
the Skinny Gibbs sampler for 1000 iterations, run the 
SOTA sampler for 100 iterations, and record the average 
time taken per iteration. All timings were obtained using a single core of 
an Apple M1 chip on a Macbook Air 2020 laptop with 16 GB RAM.

\opt{icml}{
\begin{figure}[!]
\vskip 0.2in
\begin{center}
\centerline{\includegraphics[width=1\columnwidth]{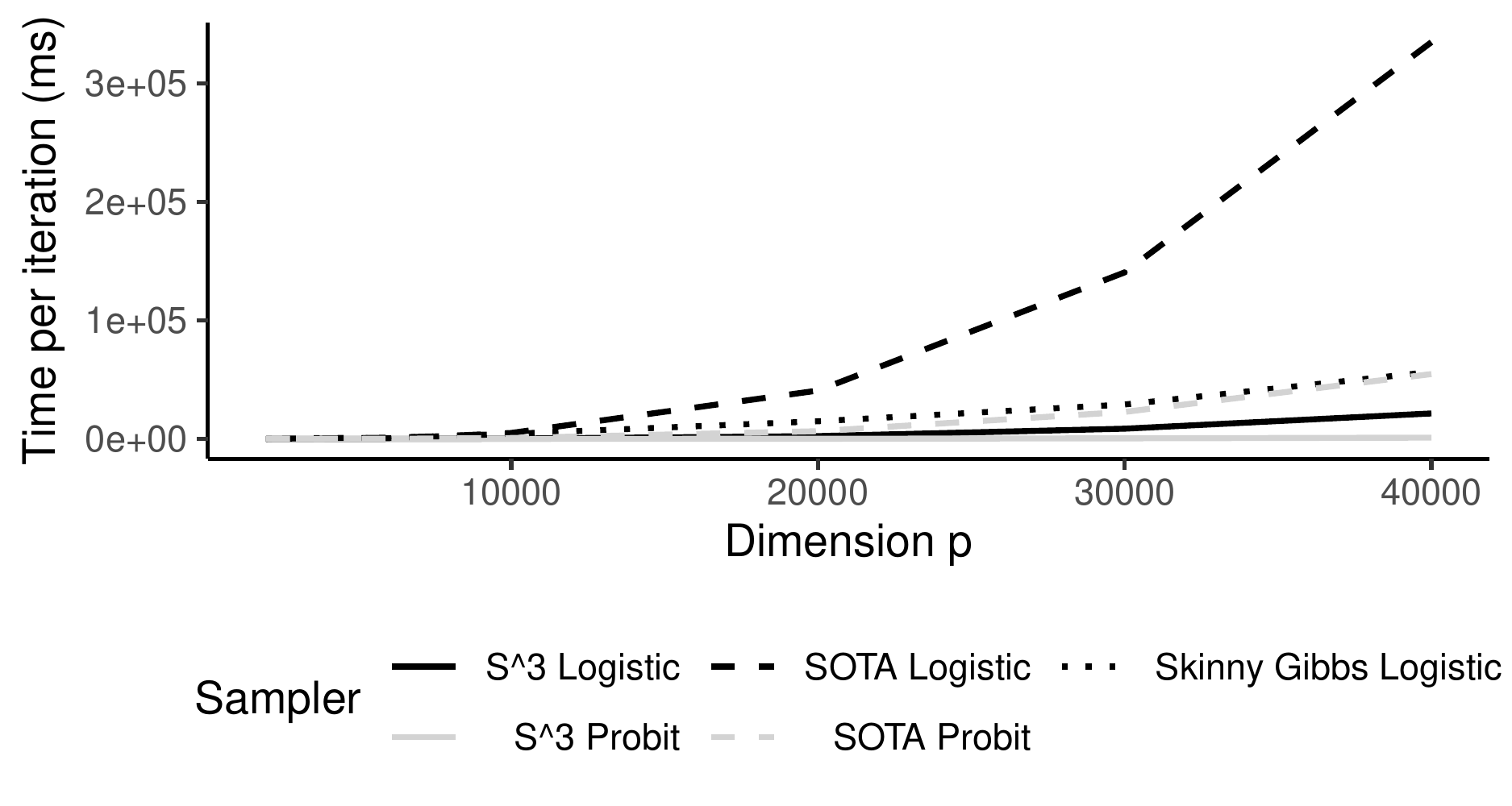}}
\vspace*{-0.05in}
\caption{Comparison of time per iteration between \sss, state-of-the-art (SOTA) 
exact MCMC sampler and the Skinny Gibbs approximate sampler of \citet{narisetty2019skinnyJASA}
on synthetic binary classification datasets. 
See Section \ref{section:comparison} for details.
}
\label{fig:time_comparison}
\end{center}
\vskip -0.2in
\end{figure}
}
\opt{arxiv}{
\begin{figure}[!]
\vskip 0.2in
\begin{center}
\centerline{\includegraphics[width=0.7\columnwidth]{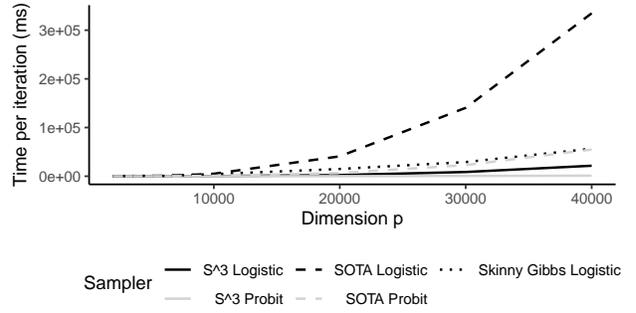}}
\vspace*{-0.05in}
\caption{Comparison of time per iteration between \sss, state-of-the-art (SOTA) 
exact MCMC sampler and the Skinny Gibbs approximate sampler of \citet{narisetty2019skinnyJASA}
on synthetic binary classification datasets. 
See Section \ref{section:comparison} for details.
}
\label{fig:time_comparison}
\end{center}
\vskip -0.2in
\end{figure}
}

Figure \ref{fig:time_comparison} 
highlights that the numerical runtimes of \sss are orders of magnitude faster than the SOTA sampler and comparable to the Skinny
Gibbs sampler. For example, for $n=4000$ 
observations, $p=40000$ covariates, and sparsity $s=400$,
\sss for logistic regression requires $21500$ms per iteration on average, which is approximately $15$ times faster than the SOTA sampler
for probit regression (which requires $335000$ms per iteration on average) and $2.5$ times faster than the Skinny Gibbs sampler 
(which requires $55600$ms per iteration on average), and \sss for probit regression requires $1100$ms per iteration on average, which is approximately $50$ times faster than the SOTA sampler
for probit regression (which requires $55800$ms per iteration on average).
For larger real-life datasets with hundreds of thousands of 
covariates, the numerical runtimes of \sss are similarly
favorable compared to the SOTA sampler. 
This is showcased in Section \ref{section:applications}, 
where for a genetics dataset, \sss is 
$50$ times faster than the SOTA sampler. 

\opt{icml}{
\begin{figure}[!]
\vskip 0.2in
\begin{center}
\centerline{\includegraphics[width=1\columnwidth]{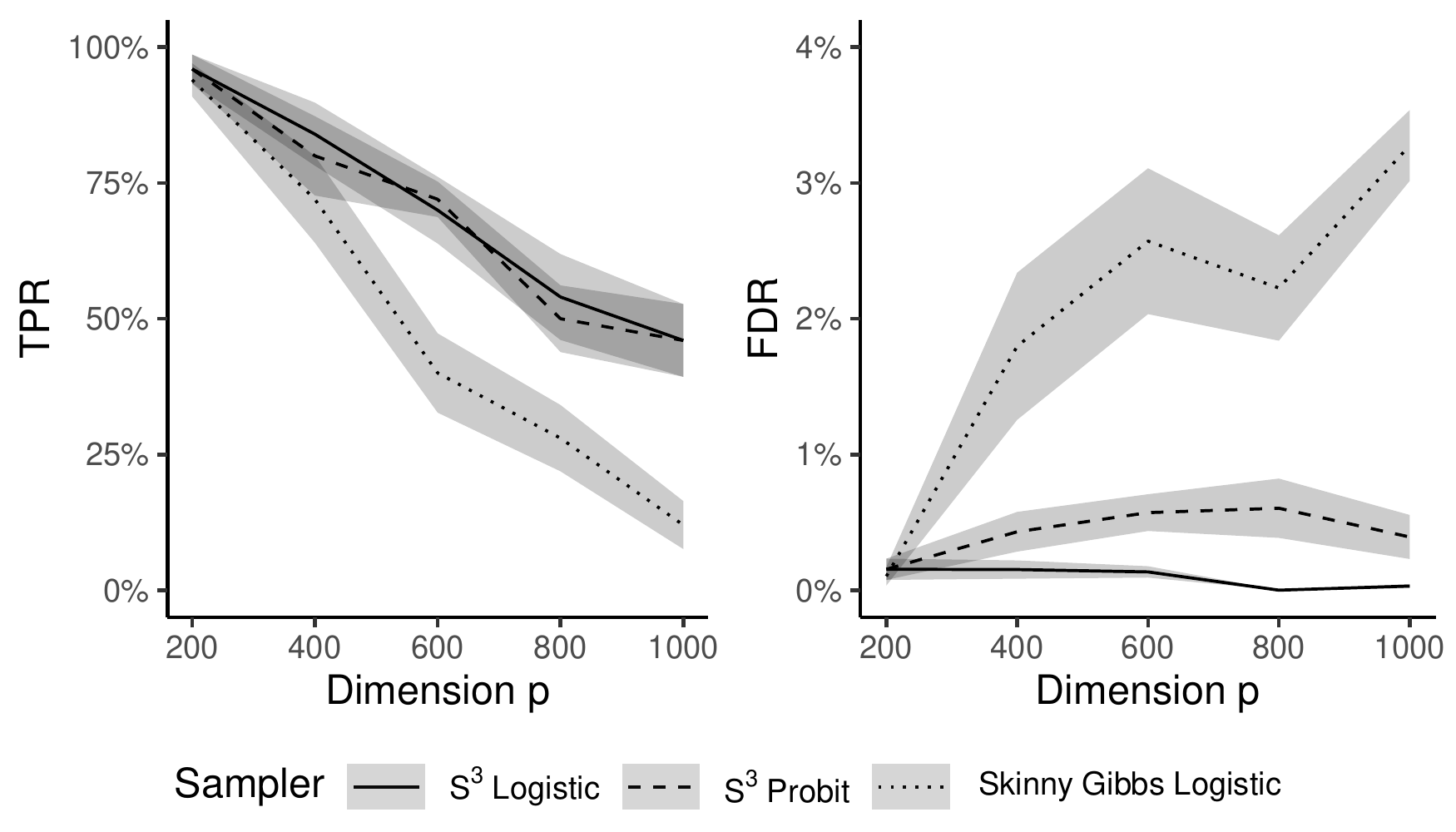}}
\vspace*{-0.05in}
\caption{
Average true positive rate (TPR) and false discovery rate (FDR) of 
\sss and the Skinny Gibbs approximate sampler \citep{narisetty2019skinnyJASA} across $20$ independently generated datasets with one standard error bars. 
See Section \ref{section:comparison} for details.
}
\label{fig:stat_comparison}
\end{center}
\vskip -0.2in
\end{figure}
}
\opt{arxiv}{
\begin{figure}[!]
\vskip 0.2in
\begin{center}
\centerline{\includegraphics[width=0.7\columnwidth]{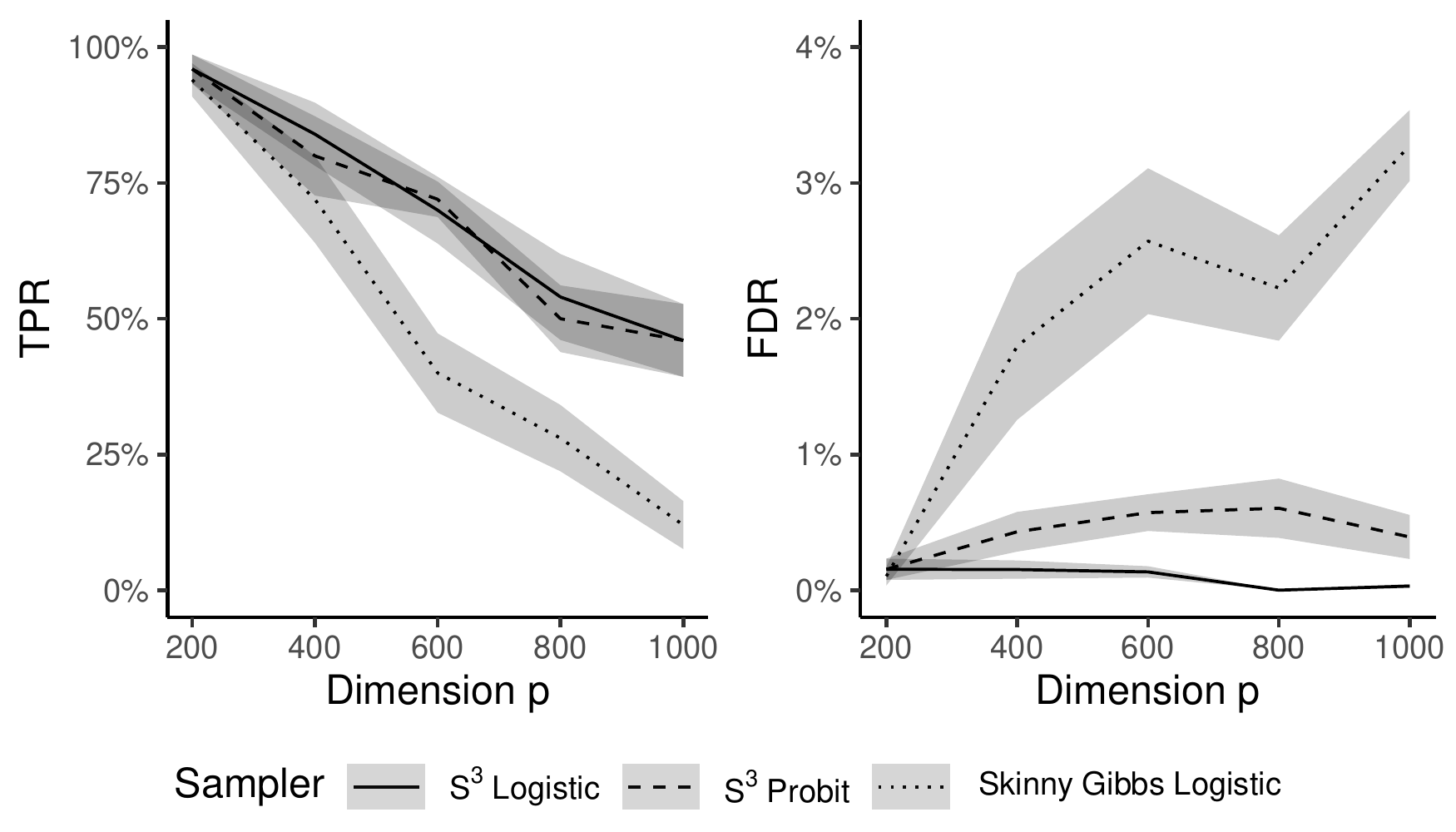}}
\vspace*{-0.05in}
\caption{
Average true positive rate (TPR) and false discovery rate (FDR) of 
\sss and the Skinny Gibbs approximate sampler \citep{narisetty2019skinnyJASA} across $20$ independently generated datasets with one standard error bars. 
See Section \ref{section:comparison} for details.
}
\label{fig:stat_comparison}
\end{center}
\vskip -0.2in
\end{figure}
}

Figure \ref{fig:stat_comparison} plots the true positive rate (TPR) 
and the false discovery rate (FDR) of variable selection based on samples 
from \sss for logistic and probit regression and Skinny Gibbs on synthetic binary classification datasets. 
To assess variable selection for signals of varying magnitude, 
we consider an exponentially decaying sparse true signal 
$\bm{\beta}^* \in \mathbb{R}^p$ such that $\beta^*_j = 2^{\frac{9-j}{4}}$
for $j \leq s$ and $\beta^*_j=0$ for $j>s$ for sparsity $s$.
The corresponding synthetically generated datasets have $n=100$
observations, varying number of covariates $p$ with $p\geq n$, and 
sparsity $s=5$. For each synthetic dataset, 
we implement \sss for logistic and probit regression 
and the Skinny Gibbs sampler for $5000$ iterations with a burn-in of $1000$ 
iterations and calculate the TPR and FDR from the samples. We use the same prior 
hyperparameters for all the algorithms, which are chosen 
according to \citet{narisetty2019skinnyJASA}. 
Additional experimental details are included in Appendices \ref{appendices:experiments}.
The SOTA sampler is not shown in Figure \ref{fig:stat_comparison}, 
as SOTA and \sss are alternative implementations of the same Gibbs sampler 
and by definition have the same statistical performance.
Figure \ref{fig:stat_comparison} shows that in 
higher dimensions, samples from \sss yield significantly higher TPR and lower FDR 
than the Skinny Gibbs sampler. We observe similar results for other choices 
of prior hyperparameters, which give \sss to either have comparable or more favorable
statistical performance to the Skinny Gibbs sampler. 

Overall, Figures \ref{fig:time_comparison} and \ref{fig:stat_comparison} highlight
that \sss can have comparable or even favorable computational cost to the 
Skinny Gibbs sampler, whilst having the correct stationary distribution and  
more favorable statistical properties in higher dimensions.
Appendix \ref{appendices:extra_experiments} contains additional simulation results 
showcasing \sss performance for individual datasets as the chain length and the total 
time elapsed varies.

\section{Applications} \label{section:applications}
We now examine the benefits of \sss on a diverse suite of regression and binary classification datasets. Table \ref{table:datasets}
summarizes the two synthetic and eight real-world datasets
considered, with further details in Appendix \ref{appendices:datasets}.

\begin{table}[t]
\caption{\label{table:datasets}
Synthetic and real-life datasets considered in
Section \ref{section:applications}. %
}
\vskip 0.15in
\begin{center}
\begin{scriptsize}
\begin{sc}
\begin{tabular}{lccc}
\toprule
Dataset & $n$ & $p$ & Response Type \\ 
\midrule
Borovecki & $31$ & $22283$ & Binary \\
Chin & $118$ & $22215$ & Binary \\
Chowdary & $104$ & $22283$ & Binary \\
Gordon & $181$ & $12533$ & Binary \\
Lymph & $148$ & $4514$ & Binary \\
Maize & $2266$ & $98385$ & Continuous \\
Malware & $373$ & $503$ & Binary  \\
PCR & $60$ & $22575$ & Continuous \\
Synthetic Binary & $1000$ & $50000$ & Continuous \\
Synthetic Continuous & $1000$ & $50000$ & Binary \\
\bottomrule
\end{tabular}
\end{sc}
\end{scriptsize}
\end{center}
\end{table}
We first consider the Gordon microarray dataset \citep{Gordon2002translationCANCER} with $n=181$ observations 
(corresponding to a binary response vector indicating presence of lung cancer)
and $p = 12533$ covariates (corresponding to genes expression levels). 
Figure \ref{fig:applications1} shows the marginal posterior 
probabilities estimated using samples from \sss and the SOTA sampler for logistic and probit regression and 
the Skinny Gibbs sampler for logistic regression, as well as the corresponding average runtimes per iteration. The marginal posterior probabilities $\pi_j\defeq \mathbb{P}_{\pi}(z_{j}=1)$ are estimated 
by $\hat\pi_j=\frac{1}{I(T-S)} \sum_{i=1}^I \sum_{t=S+1}^{T} z^{(i)}_{j,t}$, 
where $(\bm{z}_t^{(i)})_{t \geq 0}$ are samples from $i=1,\mydots,I$ independent Markov chains
generated using \sss. We sample $I=5$ independent chains of length $T=5000$ iterations with a burn-in of $1000$ iterations for both \sss and the SOTA sampler.
The average runtime per iteration with one standard error bars are calculated based on these independent chains. 

Figure \ref{fig:applications1} (Left) plots $\hat\pi_j$ against $j$ in the decreasing order of $\hat\pi_j$s. 
It shows $\hat\pi_j$s based on samples from both \sss and the SOTA sampler. 
We simulate both \sss and the SOTA sampler with the same random numbers at each iteration, 
so that any differences will be due to numerical imprecision. 
Figure \ref{fig:applications1} (Left) shows that the estimates using  \sss and the SOTA sampler are 
indistinguishable. Furthermore, in this example all components of $z_t$ are identical between \sss and the SOTA sampler
chains for all iterations $t$.
Despite producing Markov chains with indistinguishable marginal distributions and hence statistical properties, 
Figure \ref{fig:applications1} (Right) shows that \sss has approximately $20$ and $6$ times faster runtime per iteration than SOTA for logistic and probit
regression respectively. Furthermore, \sss for logistic regression has approximately $100$ times faster runtime per iteration than the Skinny Gibbs sampler. Overall, Figure \ref{fig:applications1} highlights the practical value of \sss over the SOTA sampler and the Skinny Gibbs sampler.

\begin{figure}[t]
\vskip 0.2in
\begin{center}
\centerline{\includegraphics[width=1\columnwidth]{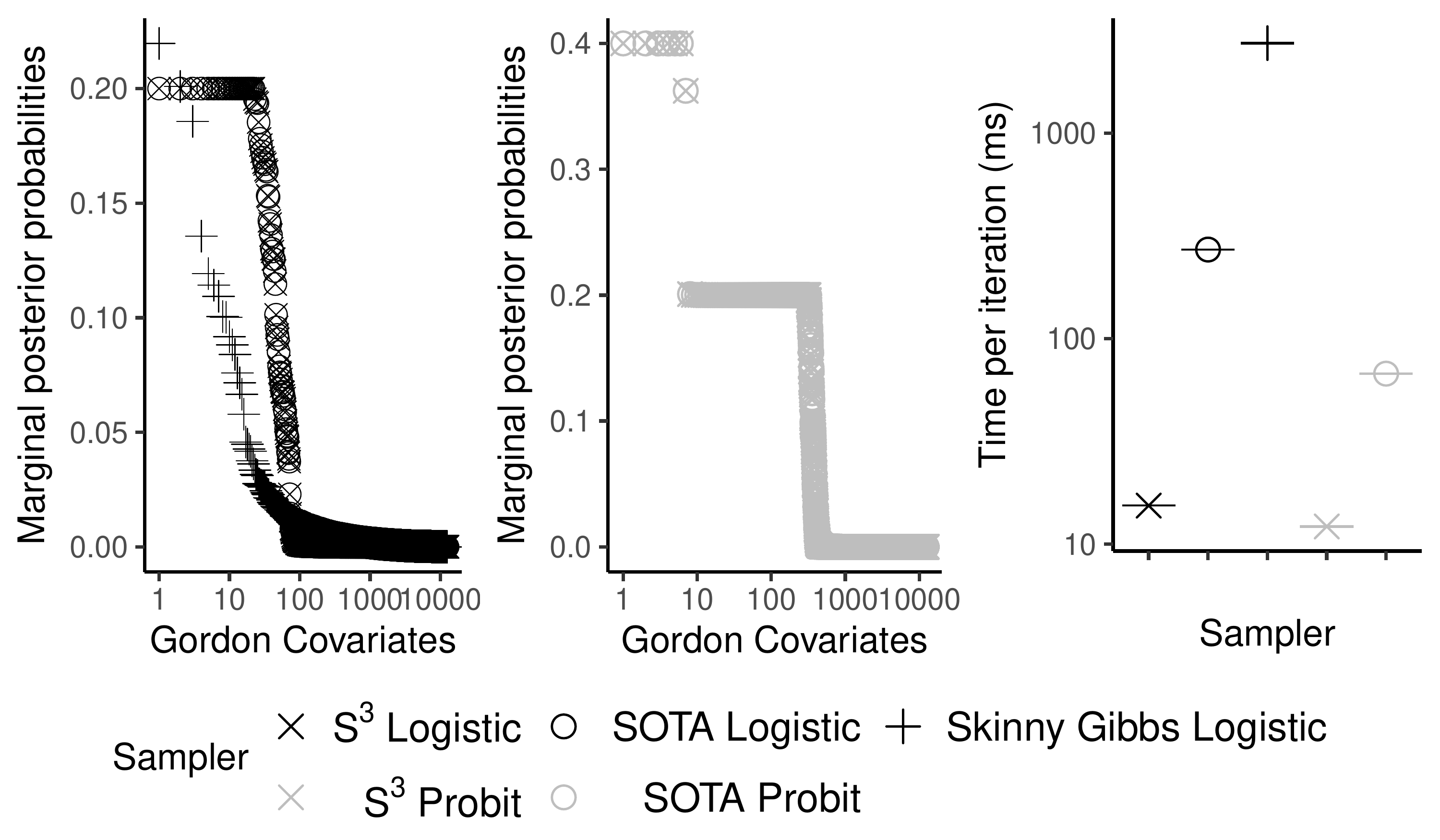}}
\vspace*{-0.05in}
\caption{Comparing Bayesian logistic and probit regression with \sss, SOTA, and Skinny Gibbs on the Gordon microarray dataset with $n=181$ observations and $p=12533$ covariates.
(\textbf{Left and Middle}) Marginal posterior probabilities $\mathbb{P}_{\pi}(z_{j}=1)$ estimated using samples from each chain: the recovered \sss and SOTA probabilities are indistinguishable but differ significantly from the Skinny Gibbs probabilities.
(\textbf{Right}) Average runtime per sampler  iteration with one standard error bars.
See Section \ref{section:applications} for details. 
}
\label{fig:applications1}
\end{center}
\vskip -0.2in
\end{figure}

Figure \ref{fig:applications2}, plotted with the y-axis on the log-scale, compares the runtimes of \sss and the SOTA sampler for linear and probit regression on ten regression and binary classification datasets respectively. 
It plots the average runtimes with one standard error bars
based on $10$ independent chains each of length $1000$ and $100$ for \sss and the SOTA sampler respectively. 
Figure \ref{fig:applications2} shows that \sss has lower runtimes per iteration compared to the SOTA sampler for all the datasets considered, with the most substantial speedups for larger datasets. For example, for the 
Maize GWAS dataset \citep{romay2013genotypingGENOMEBIO, liu2016iterativePLOSGENETICS, zeng2017nonparametricNATURECOMM}
with $n=2266$ observations (corresponding to average number of days taken for silk emergence
in different maize lines) and $p=98385$ covariates (corresponding to single nucleotide polymorphisms (SNPs) in the genome), \sss requires $650$ms per iteration on average, which is $48$ times faster than the SOTA sampler requiring $31300$ms per iteration. For researchers, such speedups can reduce algorithm runtime from days to hours, giving substantial time and computational cost savings at no compromise to inferential quality. 
Appendix \ref{appendices:extra_experiments} contains additional results of \sss applied to 
these datasets including effective sample size (ESS) calculations, marginal posterior probabilities, and 
performance under 10-fold cross-validation. 

\opt{icml}{
\begin{figure}[t]
\vskip 0.2in
\begin{center}
\centerline{\includegraphics[width=1\columnwidth]{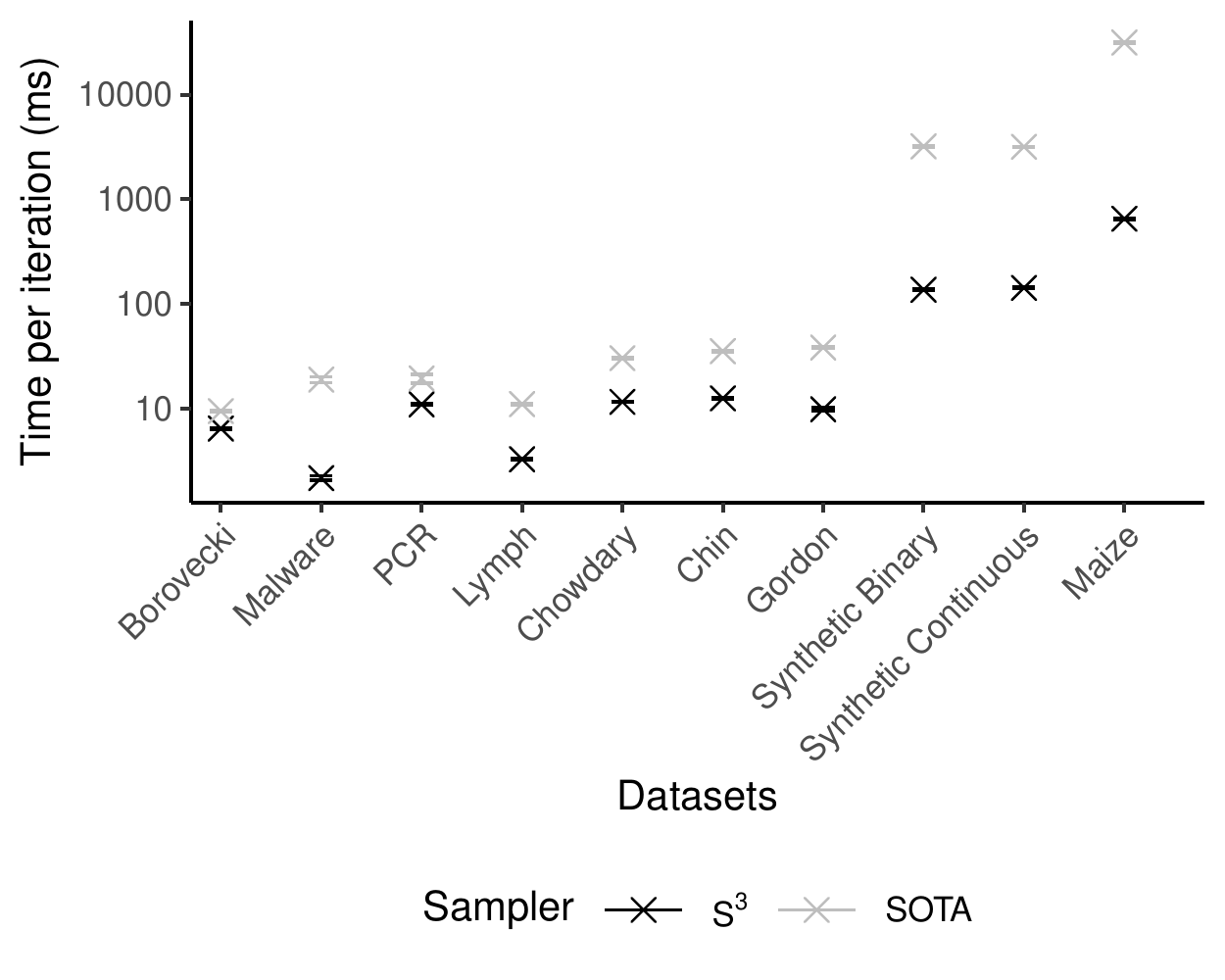}}
\vspace*{-0.05in}
\caption{
Average runtime per iteration with one standard error bars for \sss and the SOTA sampler
for linear and probit regression applied to the ten continuous and binary response datasets.
See Section \ref{section:applications} for details.}
\label{fig:applications2}
\end{center}
\vskip -0.2in
\end{figure}
}
\opt{arxiv}{
\begin{figure}[t]
\vskip 0.2in
\begin{center}
\centerline{\includegraphics[width=0.6\columnwidth]{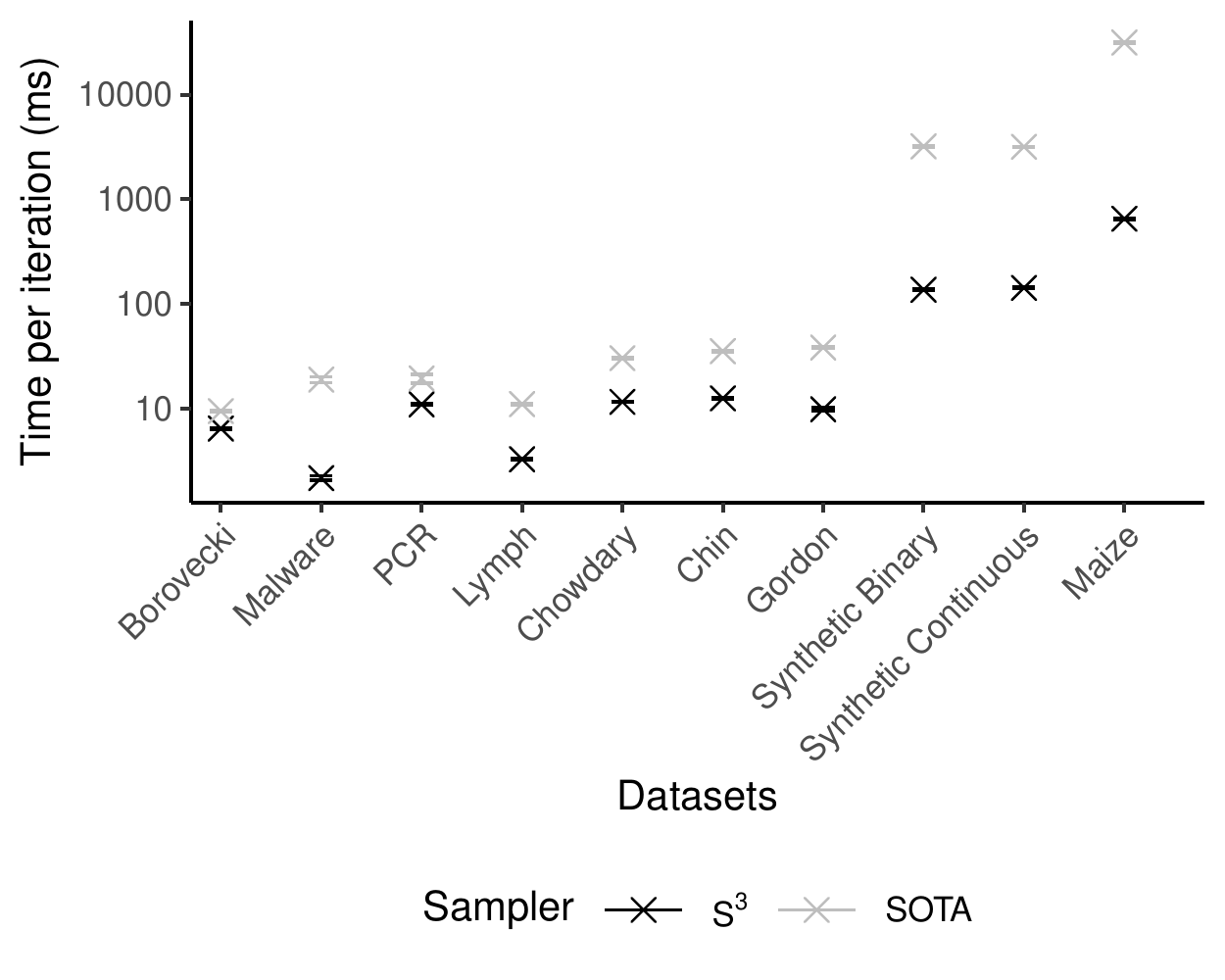}}
\vspace*{-0.05in}
\caption{
Average runtime per iteration with one standard error bars for \sss and the SOTA sampler
for linear and probit regression applied to the ten continuous and binary response datasets.
See Section \ref{section:applications} for details.}
\label{fig:applications2}
\end{center}
\vskip -0.2in
\end{figure}
}

\section{Further Work} \label{section:discussion}

The following questions arise from our work.

\textbf{Extensions of \sss to point-mass spike-and-slab 
priors, as well as to non-Gaussian tails.}
Whilst priors given in \eqref{eq:cont_spike_slab} are one of the most common formulations employed in practice, a number of alternatives are available. This includes \textit{point-mass} spike-and-slab priors 
\citep[e.g.,][]{mitchell1988bayesianJASA, johnson2012bayesianJASA}, where a degenerate Dirac distribution about zero is chosen for the spike part, 
and extensions which consider the heavier-tailed Laplace distribution 
for the slab part instead of a Gaussian distribution 
\citep{castillo2015bayesianAOS,rockova2018bayesianAOS, ray2020spikeNEURIPS, ray2021variationalJASA}. 
An extension of \sss would to be employ similar pre-computation
based strategy of Section \ref{section:mcmc} to MCMC samplers for these alternative formulations. 

\textbf{Convergence complexity analysis of \sss.}
An important question that is not addressed in this article 
is the number of iterations required for \sss or other similar samplers 
to converge to the their target posterior 
distributions. For Gibbs samplers targeting posteriors
corresponding to continuous shrinkage priors
\citep[e.g.,][]{carvalho2010theBIOMETRIKA, bhattacharya2015dirichletJASA,bhadra2019lasso}, 
much theoretical progress has been made
\citep{pal2014geometricEJS, qin2019convergenceAOS,bhattacharya2022geo,biswas2022coupled}.
Convergence of Gibbs samplers targeting spike-and-slab 
posteriors has been less extensively studied and 
requires more attention. 

\textbf{Diagnostics to assess the convergence of and 
asymptotic variance of \sss.}
Given some time and computational budget constraints, 
an immediate benefit of \sss is the ability to run longer 
Markov chains targeting spike-and-slab posteriors. This can 
alleviate some concerns linked to burn-in and asymptotic 
variance, but convergence and effective sample size diagnostics 
\citep{johnson1998coupling,biswas2019estimating,vats2020revisitingSS,vehtari2020rankBA}
remain an important consideration particularly in high-dimensional 
settings. We hope to investigate convergence diagnostics in 
future work.

\paragraph{Acknowledgments.}
We thank Juan Shen for sharing the PCR and the Lymph Node datasets, Xiaolei Liu and Xiang Zhou for sharing the Maize GWAS dataset, and Marina Vannucci for helpful feedback. NB was supported by the NSF grant DMS-1844695, a GSAS Merit Fellowship, and a Two Sigma Fellowship Award. XLM was partially supported by the NSF grant DMS-1811308.

\bibliography{references}

\begin{thebibliography}{53}
\providecommand{\natexlab}[1]{#1}
\providecommand{\url}[1]{\texttt{#1}}
\expandafter\ifx\csname urlstyle\endcsname\relax
  \providecommand{\doi}[1]{doi: #1}\else
  \providecommand{\doi}{doi: \begingroup \urlstyle{rm}\Url}\fi

\bibitem[Banerjee et~al.(2021)Banerjee, Castillo, and
  Ghosal]{banerjee2021bayesian}
Banerjee, S., Castillo, I., and Ghosal, S.
\newblock {Bayesian inference in high-dimensional models}.
\newblock \emph{Springer volume on Data Science}, 2021.

\bibitem[Barbieri \& Berger(2004)Barbieri and Berger]{barbieri2004optimalAOS}
Barbieri, M.~M. and Berger, J.~O.
\newblock {Optimal predictive model selection}.
\newblock \emph{Annals of Statistics}, 32\penalty0 (3):\penalty0 870 -- 897,
  2004.
\newblock \doi{10.1214/009053604000000238}.
\newblock URL \url{https://doi.org/10.1214/009053604000000238}.

\bibitem[Barbieri et~al.(2021)Barbieri, Berger, George, and
  Ročková]{barbieri2021medianBA}
Barbieri, M.~M., Berger, J.~O., George, E.~I., and Ročková, V.
\newblock {The Median Probability Model and Correlated Variables}.
\newblock \emph{Bayesian Analysis}, 16\penalty0 (4):\penalty0 1085 -- 1112,
  2021.
\newblock \doi{10.1214/20-BA1249}.
\newblock URL \url{https://doi.org/10.1214/20-BA1249}.

\bibitem[Bhadra et~al.(2019)Bhadra, Datta, Polson, and
  Willard]{bhadra2019lasso}
Bhadra, A., Datta, J., Polson, N.~G., and Willard, B.
\newblock {Lasso Meets Horseshoe: A Survey}.
\newblock \emph{Statistical Science}, 34\penalty0 (3):\penalty0 405 -- 427,
  2019.
\newblock \doi{10.1214/19-STS700}.
\newblock URL \url{https://doi.org/10.1214/19-STS700}.

\bibitem[Bhattacharya et~al.(2015)Bhattacharya, Pati, Pillai, and
  Dunson]{bhattacharya2015dirichletJASA}
Bhattacharya, A., Pati, D., Pillai, N.~S., and Dunson, D.~B.
\newblock {Dirichlet–Laplace Priors for Optimal Shrinkage}.
\newblock \emph{Journal of the American Statistical Association}, 110\penalty0
  (512):\penalty0 1479--1490, 2015.
\newblock \doi{10.1080/01621459.2014.960967}.
\newblock URL \url{https://doi.org/10.1080/01621459.2014.960967}.

\bibitem[Bhattacharya et~al.(2016)Bhattacharya, Chakraborty, and
  Mallick]{bhattacharya2016fastBIOMETRIKA}
Bhattacharya, A., Chakraborty, A., and Mallick, B.~K.
\newblock {Fast sampling with Gaussian scale mixture priors in high-dimensional
  regression}.
\newblock \emph{Biometrika}, 103\penalty0 (4):\penalty0 985--991, 2016.
\newblock ISSN 0006-3444.
\newblock \doi{10.1093/biomet/asw042}.
\newblock URL \url{https://doi.org/10.1093/biomet/asw042}.

\bibitem[Bhattacharya et~al.(2022)Bhattacharya, Khare, and
  Pal]{bhattacharya2022geo}
Bhattacharya, S., Khare, K., and Pal, S.
\newblock {Geometric ergodicity of Gibbs samplers for the Horseshoe and its
  regularized variants}.
\newblock \emph{Electronic Journal of Statistics}, 16\penalty0 (1):\penalty0 1
  -- 57, 2022.
\newblock \doi{10.1214/21-EJS1932}.
\newblock URL \url{https://doi.org/10.1214/21-EJS1932}.

\bibitem[Biswas et~al.(2019)Biswas, Jacob, and Vanetti]{biswas2019estimating}
Biswas, N., Jacob, P.~E., and Vanetti, P.
\newblock Estimating convergence of markov chains with l-lag couplings.
\newblock In \emph{Advances in Neural Information Processing Systems},
  volume~32. Curran Associates, Inc., 2019.
\newblock URL
  \url{https://proceedings.neurips.cc/paper/2019/file/aec851e565646f6835e915293381e20a-Paper.pdf}.

\bibitem[Biswas et~al.(2022)Biswas, Bhattacharya, Jacob, and
  Johndrow]{biswas2022coupled}
Biswas, N., Bhattacharya, A., Jacob, P.~E., and Johndrow, J.~E.
\newblock Coupling-based convergence assessment of some gibbs samplers for
  high-dimensional bayesian regression with shrinkage priors.
\newblock \emph{Journal of the Royal Statistical Society: Series B (Statistical
  Methodology)}, 2022.
\newblock \doi{10.1111/rssb.12495}.
\newblock URL \url{https://doi.org/10.1111/rssb.12495}.

\bibitem[Bogdan et~al.(2015)Bogdan, van~den Berg, Sabatti, Su, and
  Candès]{Bogdan2015slopeAOAS}
Bogdan, M., van~den Berg, E., Sabatti, C., Su, W., and Candès, E.~J.
\newblock {SLOPE—Adaptive variable selection via convex optimization}.
\newblock \emph{The Annals of Applied Statistics}, 9\penalty0 (3):\penalty0
  1103 -- 1140, 2015.
\newblock \doi{10.1214/15-AOAS842}.
\newblock URL \url{https://doi.org/10.1214/15-AOAS842}.

\bibitem[Carvalho et~al.(2010)Carvalho, Polson, and
  Scott]{carvalho2010theBIOMETRIKA}
Carvalho, C.~M., Polson, N.~G., and Scott, J.~G.
\newblock {The horseshoe estimator for sparse signals}.
\newblock \emph{Biometrika}, 97\penalty0 (2):\penalty0 465--480, 04 2010.
\newblock ISSN 0006-3444.
\newblock \doi{10.1093/biomet/asq017}.
\newblock URL \url{https://doi.org/10.1093/biomet/asq017}.

\bibitem[Castillo \& van~der Vaart(2012)Castillo and van~der
  Vaart]{castillo2012needlesAOS}
Castillo, I. and van~der Vaart, A.
\newblock {Needles and Straw in a Haystack: Posterior concentration for
  possibly sparse sequences}.
\newblock \emph{Annals of Statistics}, 40\penalty0 (4):\penalty0 2069 -- 2101,
  2012.
\newblock \doi{10.1214/12-AOS1029}.
\newblock URL \url{https://doi.org/10.1214/12-AOS1029}.

\bibitem[Castillo et~al.(2015)Castillo, Schmidt-Hieber, and van~der
  Vaart]{castillo2015bayesianAOS}
Castillo, I., Schmidt-Hieber, J., and van~der Vaart, A.
\newblock {Bayesian linear regression with sparse priors}.
\newblock \emph{Annals of Statistics}, 43\penalty0 (5):\penalty0 1986--2018,
  2015.
\newblock \doi{10.1214/15-AOS1334}.
\newblock URL \url{https://doi.org/10.1214/15-AOS1334}.

\bibitem[Dua \& Graff(2017)Dua and Graff]{Dua2019UCI}
Dua, D. and Graff, C.
\newblock {UCI} machine learning repository, 2017.
\newblock URL \url{http://archive.ics.uci.edu/ml}.

\bibitem[Flegal et~al.(2021)Flegal, Hughes, Vats, Dai, Gupta, and
  Maji]{Flegal2021mcmcse}
Flegal, J.~M., Hughes, J., Vats, D., Dai, N., Gupta, K., and Maji, U.
\newblock \emph{mcmcse: Monte Carlo Standard Errors for MCMC}.
\newblock Riverside, CA, and Kanpur, India, 2021.
\newblock R package version 1.5-0.

\bibitem[George \& McCulloch(1993)George and McCulloch]{george1993variableJASA}
George, E.~I. and McCulloch, R.~E.
\newblock {Variable Selection via Gibbs Sampling}.
\newblock \emph{Journal of the American Statistical Association}, 88\penalty0
  (423):\penalty0 881--889, 1993.
\newblock \doi{10.1080/01621459.1993.10476353}.
\newblock URL
  \url{https://www.tandfonline.com/doi/abs/10.1080/01621459.1993.10476353}.

\bibitem[George \& McCulloch(1997)George and McCulloch]{george1997approachesSS}
George, E.~I. and McCulloch, R.~E.
\newblock {Approaches for Bayesian Variable Selection}.
\newblock \emph{Statistica Sinica}, 7\penalty0 (2):\penalty0 339--373, 1997.
\newblock ISSN 10170405, 19968507.
\newblock URL \url{http://www.jstor.org/stable/24306083}.

\bibitem[Gordon et~al.(2002)Gordon, Jensen, Hsiao, Gullans, Blumenstock,
  Ramaswamy, Richards, Sugarbaker, and Bueno]{Gordon2002translationCANCER}
Gordon, G. J.~G., Jensen, R. V.~R., Hsiao, L.-L.~L., Gullans, S. R.~S.,
  Blumenstock, J. E.~J., Ramaswamy, S.~S., Richards, W. G.~W., Sugarbaker, D.
  J.~D., and Bueno, R.~R.
\newblock {Translation of Microarray Data into Clinically Relevant Cancer
  Diagnostic Tests Using Gene Expression Ratios in Lung Cancer and
  Mesothelioma}.
\newblock \emph{Cancer Research}, 62\penalty0 (17):\penalty0 4963--4967,
  September 2002.

\bibitem[Guan \& Stephens(2011)Guan and Stephens]{guan2011bayesianAOAS}
Guan, Y. and Stephens, M.
\newblock {Bayesian variable selection regression for genome-wide association
  studies and other large-scale problems}.
\newblock \emph{The Annals of Applied Statistics}, 5\penalty0 (3):\penalty0
  1780 -- 1815, 2011.
\newblock \doi{10.1214/11-AOAS455}.
\newblock URL \url{https://doi.org/10.1214/11-AOAS455}.

\bibitem[Hager(1989)]{hager1989theSIAM}
Hager, W.~W.
\newblock Updating the inverse of a matrix.
\newblock \emph{SIAM Review}, 31\penalty0 (2):\penalty0 221--239, 1989.
\newblock ISSN 00361445.
\newblock URL \url{http://www.jstor.org/stable/2030425}.

\bibitem[Hans et~al.(2007)Hans, Dobra, and West]{hans2007shotgunJASA}
Hans, C., Dobra, A., and West, M.
\newblock {Shotgun Stochastic Search for “Large p” Regression}.
\newblock \emph{Journal of the American Statistical Association}, 102\penalty0
  (478):\penalty0 507--516, 2007.
\newblock \doi{10.1198/016214507000000121}.
\newblock URL \url{https://doi.org/10.1198/016214507000000121}.

\bibitem[Held \& Holmes(2006)Held and Holmes]{held2006bayesianBA}
Held, L. and Holmes, C.~C.
\newblock {Bayesian auxiliary variable models for binary and multinomial
  regression}.
\newblock \emph{Bayesian Analysis}, 1\penalty0 (1):\penalty0 145 -- 168, 2006.
\newblock \doi{10.1214/06-BA105}.
\newblock URL \url{https://doi.org/10.1214/06-BA105}.

\bibitem[Ishwaran \& Rao(2005)Ishwaran and Rao]{ishwaran2005spikeAOS}
Ishwaran, H. and Rao, J.~S.
\newblock {Spike and slab variable selection: Frequentist and Bayesian
  strategies}.
\newblock \emph{Annals of Statistics}, 33\penalty0 (2):\penalty0 730 -- 773,
  2005.
\newblock \doi{10.1214/009053604000001147}.
\newblock URL \url{https://doi.org/10.1214/009053604000001147}.

\bibitem[Johndrow et~al.(2020)Johndrow, Orenstein, and
  Bhattacharya]{johndrow2020scalableJMLR}
Johndrow, J., Orenstein, P., and Bhattacharya, A.
\newblock {Scalable Approximate MCMC Algorithms for the Horseshoe Prior}.
\newblock \emph{Journal of Machine Learning Research}, 21\penalty0
  (73):\penalty0 1--61, 2020.
\newblock URL \url{http://jmlr.org/papers/v21/19-536.html}.

\bibitem[Johnson(1998)]{johnson1998coupling}
Johnson, V.~E.
\newblock {A coupling-regeneration scheme for diagnosing convergence in Markov
  chain Monte Carlo algorithms}.
\newblock \emph{Journal of the American Statistical Association}, 93\penalty0
  (441):\penalty0 238--248, 1998.

\bibitem[Johnson \& Rossell(2012)Johnson and Rossell]{johnson2012bayesianJASA}
Johnson, V.~E. and Rossell, D.
\newblock {Bayesian Model Selection in High-Dimensional Settings}.
\newblock \emph{Journal of the American Statistical Association}, 107\penalty0
  (498):\penalty0 649--660, 2012.
\newblock \doi{10.1080/01621459.2012.682536}.
\newblock URL \url{https://doi.org/10.1080/01621459.2012.682536}.

\bibitem[Johnstone \& Silverman(2004)Johnstone and
  Silverman]{johnstone2004needlesAOS}
Johnstone, I.~M. and Silverman, B.~W.
\newblock {Needles and straw in haystacks: Empirical Bayes estimates of
  possibly sparse sequences}.
\newblock \emph{Annals of Statistics}, 32\penalty0 (4):\penalty0 1594 -- 1649,
  2004.
\newblock \doi{10.1214/009053604000000030}.
\newblock URL \url{https://doi.org/10.1214/009053604000000030}.

\bibitem[Kelly(2007)]{kelly2007someASTRO}
Kelly, B.~C.
\newblock {Some Aspects of Measurement Error in Linear Regression of
  Astronomical Data}.
\newblock \emph{The Astrophysical Journal}, 665\penalty0 (2):\penalty0
  1489--1506, aug 2007.
\newblock \doi{10.1086/519947}.
\newblock URL \url{https://doi.org/10.1086/519947}.

\bibitem[Liang et~al.(2013)Liang, Song, and Yu]{liang2013bayesianJASA}
Liang, F., Song, Q., and Yu, K.
\newblock {Bayesian Subset Modeling for High-Dimensional Generalized Linear
  Models}.
\newblock \emph{Journal of the American Statistical Association}, 108\penalty0
  (502):\penalty0 589--606, 2013.
\newblock \doi{10.1080/01621459.2012.761942}.
\newblock URL \url{https://doi.org/10.1080/01621459.2012.761942}.

\bibitem[Liu et~al.(2016)Liu, Huang, Fan, Buckler, and
  Zhang]{liu2016iterativePLOSGENETICS}
Liu, X., Huang, M., Fan, B., Buckler, E.~S., and Zhang, Z.
\newblock {Iterative Usage of Fixed and Random Effect Models for Powerful and
  Efficient Genome-Wide Association Studies}.
\newblock \emph{PLOS Genetics}, 12\penalty0 (2):\penalty0 1--24, 2016.
\newblock \doi{10.1371/journal.pgen.1005767}.
\newblock URL \url{https://doi.org/10.1371/journal.pgen.1005767}.

\bibitem[Mitchell \& Beauchamp(1988)Mitchell and
  Beauchamp]{mitchell1988bayesianJASA}
Mitchell, T.~J. and Beauchamp, J.~J.
\newblock {Bayesian Variable Selection in Linear Regression}.
\newblock \emph{Journal of the American Statistical Association}, 83\penalty0
  (404):\penalty0 1023--1032, 1988.
\newblock ISSN 01621459.
\newblock URL \url{http://www.jstor.org/stable/2290129}.

\bibitem[Narisetty \& He(2014)Narisetty and He]{narisetty2014bayesianAOS}
Narisetty, N.~N. and He, X.
\newblock {Bayesian variable selection with shrinking and diffusing priors}.
\newblock \emph{Annals of Statistics}, 42\penalty0 (2):\penalty0 789 -- 817,
  2014.
\newblock \doi{10.1214/14-AOS1207}.
\newblock URL \url{https://doi.org/10.1214/14-AOS1207}.

\bibitem[Narisetty et~al.(2019)Narisetty, Shen, and
  He]{narisetty2019skinnyJASA}
Narisetty, N.~N., Shen, J., and He, X.
\newblock {Skinny Gibbs: A Consistent and Scalable Gibbs Sampler for Model
  Selection}.
\newblock \emph{Journal of the American Statistical Association}, 114\penalty0
  (527):\penalty0 1205--1217, 2019.
\newblock \doi{10.1080/01621459.2018.1482754}.
\newblock URL \url{https://doi.org/10.1080/01621459.2018.1482754}.

\bibitem[O'Brien \& Dunson(2004)O'Brien and
  Dunson]{obrien2004bayesianBIOMETRICS}
O'Brien, S.~M. and Dunson, D.~B.
\newblock {Bayesian multivariate logistic regression}.
\newblock \emph{Biometrics}, 60\penalty0 (3):\penalty0 739--746, 2004.
\newblock \doi{10.1111/j.0006-341X.2004.00224.x}.
\newblock URL \url{https://doi:10.1111/j.0006-341X.2004.00224.x}.

\bibitem[Pal \& Khare(2014)Pal and Khare]{pal2014geometricEJS}
Pal, S. and Khare, K.
\newblock {Geometric ergodicity for Bayesian shrinkage models}.
\newblock \emph{Electronic Journal of Statistics}, 8\penalty0 (1):\penalty0
  604--645, 2014.
\newblock \doi{10.1214/14-EJS896}.
\newblock URL \url{https://doi.org/10.1214/14-EJS896}.

\bibitem[Polson et~al.(2013)Polson, Scott, and Windle]{polson2013bayesianJASA}
Polson, N.~G., Scott, J.~G., and Windle, J.
\newblock {Bayesian Inference for Logistic Models Using Pólya–Gamma Latent
  Variables}.
\newblock \emph{Journal of the American Statistical Association}, 108\penalty0
  (504):\penalty0 1339--1349, 2013.
\newblock \doi{10.1080/01621459.2013.829001}.
\newblock URL \url{https://doi.org/10.1080/01621459.2013.829001}.

\bibitem[Qin \& Hobert(2019)Qin and Hobert]{qin2019convergenceAOS}
Qin, Q. and Hobert, J.~P.
\newblock {Convergence complexity analysis of Albert and Chib’s algorithm for
  Bayesian probit regression}.
\newblock \emph{Annals of Statistics}, 47\penalty0 (4):\penalty0 2320--2347,
  2019.
\newblock \doi{10.1214/18-AOS1749}.
\newblock URL \url{https://doi.org/10.1214/18-AOS1749}.

\bibitem[Ray \& Szabó(2021)Ray and Szabó]{ray2021variationalJASA}
Ray, K. and Szabó, B.
\newblock Variational bayes for high-dimensional linear regression with sparse
  priors.
\newblock \emph{Journal of the American Statistical Association}, 0\penalty0
  (0):\penalty0 1--12, 2021.
\newblock \doi{10.1080/01621459.2020.1847121}.
\newblock URL \url{https://doi.org/10.1080/01621459.2020.1847121}.

\bibitem[Ray et~al.(2020)Ray, Szabo, and Clara]{ray2020spikeNEURIPS}
Ray, K., Szabo, B., and Clara, G.
\newblock Spike and slab variational bayes for high dimensional logistic
  regression.
\newblock In \emph{Advances in Neural Information Processing Systems},
  volume~33, pp.\  14423--14434. Curran Associates, Inc., 2020.
\newblock URL
  \url{https://proceedings.neurips.cc/paper/2020/file/a5bad363fc47f424ddf5091c8471480a-Paper.pdf}.

\bibitem[Romay et~al.(2013)Romay, Millard, Glaubitz, Peiffer, Swarts,
  Casstevens, Elshire, Acharya, Mitchell, Flint-Garcia, McMullen, Holland,
  Buckler, and Gardner]{romay2013genotypingGENOMEBIO}
Romay, M.~C., Millard, M.~J., Glaubitz, J.~C., Peiffer, J.~A., Swarts, K.~L.,
  Casstevens, T.~M., Elshire, R.~J., Acharya, C.~B., Mitchell, S.~E.,
  Flint-Garcia, S.~A., McMullen, M.~D., Holland, J.~B., Buckler, E.~S., and
  Gardner, C.~A.
\newblock {Comprehensive genotyping of the USA national maize inbred seed
  bank}.
\newblock \emph{Genome Biology}, 14\penalty0 (6):\penalty0 R55, 2013.
\newblock \doi{10.1186/gb-2013-14-6-r55}.
\newblock URL \url{https://doi.org/10.1186/gb-2013-14-6-r55}.

\bibitem[Ročková(2018)]{rockova2018bayesianAOS}
Ročková, V.
\newblock {Bayesian estimation of sparse signals with a continuous
  spike-and-slab prior}.
\newblock \emph{Annals of Statistics}, 46\penalty0 (1):\penalty0 401 -- 437,
  2018.
\newblock \doi{10.1214/17-AOS1554}.
\newblock URL \url{https://doi.org/10.1214/17-AOS1554}.

\bibitem[Scott \& Berger(2010)Scott and Berger]{scott2010bayesAOS}
Scott, J.~G. and Berger, J.~O.
\newblock {Bayes and empirical-Bayes multiplicity adjustment in the
  variable-selection problem}.
\newblock \emph{Annals of Statistics}, 38\penalty0 (5):\penalty0 2587 -- 2619,
  2010.
\newblock \doi{10.1214/10-AOS792}.
\newblock URL \url{https://doi.org/10.1214/10-AOS792}.

\bibitem[Sereno(2015)]{sereno2015bayesianMONTHLY}
Sereno, M.
\newblock {A Bayesian approach to linear regression in astronomy}.
\newblock \emph{Monthly Notices of the Royal Astronomical Society},
  455\penalty0 (2):\penalty0 2149--2162, 11 2015.
\newblock ISSN 0035-8711.
\newblock \doi{10.1093/mnras/stv2374}.
\newblock URL \url{https://doi.org/10.1093/mnras/stv2374}.

\bibitem[Tadesse \& Vannucci(2021)Tadesse and Vannucci]{vannucci2021handbook}
Tadesse, M.~G. and Vannucci, M.
\newblock \emph{{Handbook of Bayesian Variable Selection}}.
\newblock Chapman and Hall/CRC, 2021.
\newblock \doi{10.1201/9781003089018}.
\newblock URL \url{https://doi.org/10.1201/9781003089018}.

\bibitem[Tibshirani(1996)]{tibshirani1994regression}
Tibshirani, R.
\newblock {Regression Shrinkage and Selection Via the Lasso}.
\newblock \emph{Journal of the Royal Statistical Society: Series B (Statistical
  Methodology)}, 58\penalty0 (1):\penalty0 267--288, 1996.
\newblock \doi{10.1111/j.2517-6161.1996.tb02080.x}.
\newblock URL \url{https://doi.org/10.1111/j.2517-6161.1996.tb02080.x}.

\bibitem[Titsias \& L\'{a}zaro-Gredilla(2011)Titsias and
  L\'{a}zaro-Gredilla]{titsias2011spikeNEURIPS}
Titsias, M. and L\'{a}zaro-Gredilla, M.
\newblock Spike and slab variational inference for multi-task and multiple
  kernel learning.
\newblock In \emph{Advances in Neural Information Processing Systems},
  volume~24. Curran Associates, Inc., 2011.
\newblock URL
  \url{https://proceedings.neurips.cc/paper/2011/file/b495ce63ede0f4efc9eec62cb947c162-Paper.pdf}.

\bibitem[Vats \& Knudson(2021)Vats and Knudson]{vats2020revisitingSS}
Vats, D. and Knudson, C.
\newblock {Revisiting the Gelman–Rubin Diagnostic}.
\newblock \emph{Statistical Science}, 36\penalty0 (4):\penalty0 518 -- 529,
  2021.
\newblock \doi{10.1214/20-STS812}.
\newblock URL \url{https://doi.org/10.1214/20-STS812}.

\bibitem[Vats et~al.(2019)Vats, Flegal, and
  Jones]{vats2020multivariateBIOMETRIKA}
Vats, D., Flegal, J.~M., and Jones, G.~L.
\newblock {Multivariate output analysis for Markov chain Monte Carlo}.
\newblock \emph{Biometrika}, 106\penalty0 (2):\penalty0 321--337, 2019.
\newblock ISSN 0006-3444.
\newblock \doi{10.1093/biomet/asz002}.
\newblock URL \url{https://doi.org/10.1093/biomet/asz002}.

\bibitem[Vehtari et~al.(2021)Vehtari, Gelman, Simpson, Carpenter, and
  Bürkner]{vehtari2020rankBA}
Vehtari, A., Gelman, A., Simpson, D., Carpenter, B., and Bürkner, P.-C.
\newblock {Rank-Normalization, Folding, and Localization: An Improved
  $\widehat{R}$ for Assessing Convergence of MCMC (with Discussion)}.
\newblock \emph{Bayesian Analysis}, 16\penalty0 (2):\penalty0 667 -- 718, 2021.
\newblock \doi{10.1214/20-BA1221}.
\newblock URL \url{https://doi.org/10.1214/20-BA1221}.

\bibitem[Yang et~al.(2016)Yang, Wainwright, and Jordan]{yang2016onAOS}
Yang, Y., Wainwright, M.~J., and Jordan, M.~I.
\newblock {On the computational complexity of high-dimensional Bayesian
  variable selection}.
\newblock \emph{Annals of Statistics}, 44\penalty0 (6):\penalty0 2497--2532,
  2016.
\newblock \doi{10.1214/15-AOS1417}.
\newblock URL \url{https://doi.org/10.1214/15-AOS1417}.

\bibitem[Zeng \& Zhou(2017)Zeng and Zhou]{zeng2017nonparametricNATURECOMM}
Zeng, P. and Zhou, X.
\newblock {Non-parametric genetic prediction of complex traits with latent
  Dirichlet process regression models}.
\newblock \emph{Nature Communications}, 8\penalty0 (1):\penalty0 456, 2017.
\newblock \doi{10.1038/s41467-017-00470-2}.
\newblock URL \url{https://doi.org/10.1038/s41467-017-00470-2}.

\bibitem[Zhou et~al.(2013)Zhou, Carbonetto, and
  Stephens]{zhou2013polygenicPLOS}
Zhou, X., Carbonetto, P., and Stephens, M.
\newblock {Polygenic Modeling with {B}ayesian Sparse Linear Mixed Models}.
\newblock \emph{PLOS Genetics}, 9\penalty0 (2):\penalty0 1--14, 02 2013.
\newblock \doi{10.1371/journal.pgen.1003264}.
\newblock URL \url{https://doi.org/10.1371/journal.pgen.1003264}.

\bibitem[Zou \& Hastie(2005)Zou and Hastie]{zou2005regularization}
Zou, H. and Hastie, T.
\newblock {Regularization and variable selection via the elastic net}.
\newblock \emph{Journal of the Royal Statistical Society: Series B (Statistical
  Methodology)}, 67\penalty0 (2):\penalty0 301--320, 2005.
\newblock \doi{10.1111/j.1467-9868.2005.00503.x}.
\newblock URL \url{https://doi.org/10.1111/j.1467-9868.2005.00503.x}.

\end{thebibliography}
\opt{icml}{\bibliographystyle{icml2022_template/icml2022}}

\appendix
\onecolumn
\section{Proofs} \label{appendices:proofs}
\begin{proof}[Proof of Proposition \ref{prop:comp_cost}]
Consider Step $1$ of Algorithm \ref{algo:linear_spike_slab} and
Algorithm \ref{algo:logistic_probit_spike_slab} for probit regression. Given pre-computed matrices $\bm{\tilde{M}}_{\tau_0}$, 
$\bm{\tilde{M}}_{\tau_0}$, $\bm{M}_{t-1}$, $\bm{M}_{t-1}^{-1}$ 
and state $\bm{z}_{t-1}$, calculating $\bm{M}_{t}$, $\bm{M}_{t}^{-1}$ requires 
$\mathcal{O}(n^2 p_t)$ cost
by \eqref{eq:pre_compute_correction1} and \eqref{eq:pre_compute_correction2}, 
where $p_t = \min \{ \|\bm{z}_t\|_1, p-\|\bm{z}_t\|_1, \delta_t \}$ for $\delta_t = \|\bm{z}_t - \bm{z}_{t-1}\|_1$. 

Consider Step $1$ of Algorithm \ref{algo:logistic_probit_spike_slab} for logistic regression. This requires $\mathcal{O}( \max\{n^2 p_t, n^3\} )$ cost, where the $\mathcal{O}(n^2p_t)$ cost arises from the calculation of $\bm{M}_t$ using \eqref{eq:pre_compute_correction3} and the $\mathcal{O}(n^3)$ cost arises from inverting $\bm{M}_t$ to calculate $\bm{M}_t^{-1}$. 

Given $\bm{M}_{t}^{-1}$, Step $2$ of Algorithms \ref{algo:linear_spike_slab} and
\ref{algo:logistic_probit_spike_slab} then requires $\mathcal{O}(n p)$ cost, which arises from the matrix vector product $\bm{X}\bm{D}_t^{-\frac{1}{2}}r$ for $r \sim \mathcal{N}(0, \bm{I}_p)$ in Algorithm \ref{algo:fast_mvn_bhattacharya}. 
By component-wise independence, Step $3$ of Algorithms \ref{algo:linear_spike_slab} and
\ref{algo:logistic_probit_spike_slab} costs $\mathcal{O}(p)$ and 
Step $4$ of Algorithm \ref{algo:logistic_probit_spike_slab} cost $\mathcal{O}(n)$. 
Step $4$ of Algorithm \ref{algo:logistic_probit_spike_slab} and
Step $5$ of Algorithm \ref{algo:logistic_probit_spike_slab} for probit regression both cost $\mathcal{O}(1)$.
Step $5$ of Algorithm \ref{algo:logistic_probit_spike_slab} for logistic regression both costs $\mathcal{O}(n)$.

This gives an overall computational cost of $\mathcal{O}(\max \{n^2 p_t, np \})$ for
Algorithm \ref{algo:linear_spike_slab} and Algorithm \ref{algo:logistic_probit_spike_slab} for probit regression at 
iteration $t$, and a cost of $\mathcal{O}(\max \{n^2 p_t, n^3, np \})$ for
\ref{algo:logistic_probit_spike_slab} for logistic regression. 
\end{proof}

\begin{proof}[Proof of Proposition \ref{prop:expected_comp_cost}]
By linearity, $\mathbb{E}[\delta_t] = \sum_{j=1}^p \mathbb{P}(z_{j,t} \neq z_{j,t-1})$
where for each component $j$ the random variables $z_{j,t}$ and $z_{j,t-1}$ are on $\{0,1\}$. For each $j$, we obtain
\begin{flalign*}
\mathbb{P}(z_{j,t} \neq z_{j,t-1}) &= \mathbb{P}(z_{j,t}=1, z_{j,t-1}=0) + \mathbb{P}(z_{j,t}=0, z_{j,t-1}=1) \\
&= \big( \mathbb{P}(z_{j,t}=1) - \mathbb{P}(z_{j,t}=1, z_{j,t-1}=1) \big) + 
\big( \mathbb{P}(z_{j,t-1}=1) - \mathbb{P}(z_{j,t}=1, z_{j,t-1}=1) \big) \\
&= \big( \mathbb{P}(z_{j,t}=1) - \mathrm{cov}(z_{j,t}, z_{j,t-1}) - \mathbb{P}(z_{j,t}=1)\mathbb{P}(z_{j,t-1}=1) \big) + \\
& \quad \quad \big( \mathbb{P}(z_{j,t-1}=1) - \mathrm{cov}(z_{j,t}, z_{j,t-1}) - \mathbb{P}(z_{j,t}=1)\mathbb{P}(z_{j,t-1}=1) \big) \\
&= \mathbb{P}(z_{j,t}=1)\mathbb{P}(z_{j,t-1}=0) + \mathbb{P}(z_{j,t}=0)\mathbb{P}(z_{j,t-1}=1)
- 2 \mathrm{cov}(z_{j,t}, z_{j,t-1}).
\end{flalign*}
When $z_{j, t-1}$ follows the stationary $\pi$, $z_{j,t} \sim z_{j, t-1}$ and $\mathrm{var}(z_{j,t}) = \mathbb{P}(z_{j,t}=1)\mathbb{P}(z_{j,t-1}=0)$. Consequently,
\begin{flalign*}
\mathbb{P}(z_{j,t} \neq z_{j,t-1}) = 2 \mathrm{var}_\pi(z_{j,t}) - 2 \mathrm{cov}(z_{j,t}, z_{j,t-1}) = 2 \mathrm{var}_\pi(z_{j,t})(1 - \mathrm{corr}_\pi(z_{j,t}, z_{j,t-1})).
\end{flalign*}
\end{proof}

\begin{proof}[Proof of Proposition \ref{prop:pt_empirical_bound}]
Note that $a^2=a$ if $a$ only takes the value $0$ and $1$. This gives
\begin{equation}\label{eq:identity}
{\rm I}\{z_{j, t}\not= z_{j, t-1}\}=(z_{j, t}-z_{j, t-1})^2=z_{j, t}+z_{j, t-1}-2z_{j, t}z_{j, t-1}.
\end{equation}
Let $J$ be the uniform random variable on the integers $\{1, \ldots, p\}$. Then, 
\begin{equation}\label{eq:ptexact}
    \delta_t = p \big( \EJ(z_{J, t})+\EJ(z_{J, t-1}) - 2 \EJ(z_{J, t}z_{J, t-1}) \big),
\end{equation}
where the expectation is taken with respect to the random index $J$. Note that 
$\EJ(z_{J, t})=\|\bm{z}_t\|_1/p$ and $\VJ(z_{J, t})=(\|\bm{z}_t\|_1/p)(1-\|\bm{z}_t\|_1/p)=\tau_t^2/p^2$.
We obtain 
\begin{align}
    \EJ(z_{J, t}z_{J, t-1}) = \CJ(z_{J, t}, z_{J, t-1})+ \EJ(z_{J, t})\EJ(z_{J, t-1})
    = \frac{\rho_t \tau_t \tau_{t-1} + \|\bm{z}_t\|_1 \|\bm{z}_{t-1}\|_1}{p^2},\label{eq:etexact}
\end{align}
where $\rho_t=\cJ(z_{J, t}, z_{J, t-1})$. Combining \eqref{eq:ptexact}-\eqref{eq:etexact} yields \eqref{eq:exactp}. 
\end{proof}

\section{Algorithm Derivations} \label{appendices:algo_derivations}
\subsection{Linear regression with spike-and-slab priors}
For linear regression with the continuous spike-and-slab priors in \eqref{eq:cont_spike_slab}, 
the posterior density of $(\bm{\beta}, \bm{z}, \sigma^2) \in \mathbb{R}^p \times \{0,1\}^p \times (0,\infty)$ is given by
\begin{flalign} \label{eq:linear_reg_posterior}
\pi(\bm{\beta}, \bm{z}, \sigma^2 | \bm{y}) \propto &  \mathcal{N} (\bm{y}; \bm{X}\bm{\beta}, \sigma^2) \InvGamma \Big(\sigma^2 ; \frac{a_0}{2}, \frac{b_0}{2} \Big) \\
&\prod_{j=1}^p \big( q \mathcal{N}(\beta_j; 0, \sigma^2 \tau_1^2) \big)^{z_j} \big((1-q) \mathcal{N}(\beta_j; 0, \sigma^2 \tau_0^2)\big)^{1-z_j}.
\end{flalign}
From \eqref{eq:linear_reg_posterior}, we can calculate the conditional distributions. We have
\begin{flalign*}
\pi(\bm{\beta} | \bm{z}, \sigma^2, \bm{y}) &\propto  \mathcal{N} (\bm{y}; \bm{X}\bm{\beta}, \sigma^2 \bm{I}_n) \prod_{j=1}^p \mathcal{N}(\beta_j; 0, \sigma^2 \tau_1^2)^{z_j}  \mathcal{N}(\beta_j; 0, \sigma^2 \tau_0^2)^{1-z_j} \\
&\propto  \mathcal{N} (\bm{y}; \bm{X}\bm{\beta}, \sigma^2 \bm{I}_n) \mathcal{N}(\bm{\beta}; 0, \sigma^2 \bm{D}^{-1}) \text{ for } \bm{D} \defeq \Diag(\bm{z} \tau_1^{-2} + (\bm{\mathrm{1}}_p-\bm{z}) \tau_0^{-2}) \\
&\propto \exp \left\{ -\frac{1}{2\sigma^2 \bm{I}_n} \big( \bm{\beta}^\top \bm{X}^\top \bm{X} \bm{\beta} - 2\bm{\beta}^\top \bm{X}^\top \bm{y} + \bm{\beta}^\top \bm{D} \bm{\beta} \big) \right\}  \\
&\propto \mathcal{N}(\bm{\beta}; \bm{\Sigma}^{-1} \bm{X}^\top \bm{y}, \sigma^2 \bm{\Sigma}^{-1}) \text{ for } \bm{\Sigma} = \bm{X}^\top \bm{X} + \bm{D}, \\
\pi(\bm{z} | \bm{\beta}, \sigma^2, \bm{y}) &\propto \prod_{j=1}^p \big( q \mathcal{N}(\beta_j; 0, \sigma^2 \tau_1^2) \big)^{z_j} \big((1-q) \mathcal{N}(\beta_j; 0, \sigma^2 \tau_0^2)\big)^{1-z_j} \\
&\propto \prod_{j=1}^p \Bernoulli \Big( z_i; \frac{q \mathcal{N}(\beta_j; 0, \sigma^2 \tau_1^2)}{q \mathcal{N}(\beta_j; 0, \sigma^2 \tau_1^2) + (1-q) \mathcal{N}(\beta_j; 0, \sigma^2 \tau_0^2)} \Big), \text{ and } \\
\pi(\sigma^2 | \bm{\beta}, \bm{z}, \bm{y}) &\propto \mathcal{N} (\bm{y}; \bm{X}\bm{\beta}, \sigma^2) \InvGamma \Big(\sigma^2 ; \frac{a_0}{2}, \frac{b_0}{2} \Big) \mathcal{N} (\bm{\beta} ; 0, \sigma^2 \bm{D}^{-1} ) \\
&\propto \Big(\frac{1}{\sigma^2}\Big)^{\frac{n}{2}} \exp\left\{ -\frac{1}{2 \sigma^2} \| \bm{y} - \bm{X} \bm{\beta} \|_2^2 \right\} \Big( \frac{1}{\sigma^2}\Big)^{\frac{a_0}{2}} \exp \left\{ -\frac{1}{2 \sigma^2} b_0 \right\} \Big(\frac{1}{\sigma^2}\Big)^{\frac{p}{2}} \exp \left\{-\frac{1}{2 \sigma^2} \bm{\beta}^\top \bm{D} \bm{\beta} \right\} \\
&\propto \InvGamma \Big( \sigma^2 ; \frac{a_0+n+p}{2}, \frac{b_0 + \| \bm{y} - \bm{X} \bm{\beta} \|_2^2 + \bm{\beta}^\top \bm{D} \bm{\beta}}{2} \Big)
\end{flalign*}
as given in Algorithm \ref{algo:linear_spike_slab}.

\subsection{Probit regression with spike-and-slab priors} \label{appendices:probit}
Consider the probit regression likelihood, where for each observation $i=1,\mydots,n$,
$\mathbb{P}( y_i = 1 | \bm{x_i}, \bm{\beta}) = 1 - \mathbb{P}( y_i = 0 | \bm{x_i}, \bm{\beta}) =
\Phi( \bm{x_i}^\top \bm{\beta} )$
for $\bm{x_i}^\top$ the $i$-th row of the design matrix $\bm{X}$.
and $\Phi$ that cumulative density function of a univariate Normal distribution.  
We obtain $y_i = \mathrm{I} \{ \tilde{y}_i>0 \}$ for 
$\tilde{y}_i | \bm{\beta} \sim \mathcal{N}(\bm{x_i}^\top \bm{\beta}, 1)$.
The Bayesian probit regression model is then given by 
\begin{flalign}
z_j &\overset{i.i.d.}{\sim} \Bernoulli(q) \quad \text{ for all } j=1,\mydots,p \nonumber \\
\beta_j | z_j &\overset{ \ ind \ }{\sim} (1-z_j) \mathcal{N}(0, \tau_0^2) + z_j \mathcal{N}(0, \tau_1^2) 
\quad \text{ for all } j=1,\mydots,p \label{eq:probit_spike_slab}  \\
\tilde{y}_i | \bm{\beta} &\overset{ \ ind \ }{\sim} \mathcal{N}(\bm{x_i}^\top \bm{\beta}, 1) 
\quad \text{ for all } i=1,\mydots,n \nonumber \\
y_i &\overset{ \ \ \quad \ \ }{=} \mathrm{I}\{ \tilde{y}_i > 0 \} \quad \text{ for all } i=1,\mydots,n.\nonumber
\end{flalign}
For the prior and likelihood in \eqref{eq:probit_spike_slab}, the posterior density 
of $(\bm{\beta}, \bm{z}, \bm{\tilde{y}}) \in \mathbb{R}^p \times \{0,1\}^p \times \mathbb{R}^p$ is given by
\begin{flalign}
\pi(\bm{\beta}, \bm{z}, \bm{\tilde{y}} | \bm{y}) \propto & \Big( \prod_{i=1}^n \mathrm{I} \big\{ \mathrm{I} \{\tilde{y}_i>0\} = y_i \big\} \mathcal{N} (\tilde{y}_i; \bm{x_i}^\top \bm{\beta}, 1) \Big) 
\nonumber \\ 
& \prod_{j=1}^p \big( q \mathcal{N}(\beta_j; 0, \tau_1^2) \big)^{z_j} \big((1-q) \mathcal{N}(\beta_j; 0, \tau_0^2)\big)^{1-z_j}.
\label{eq:probit_reg_posterior}
\end{flalign}
From \eqref{eq:probit_reg_posterior}, we can calculate the conditional distributions. We obtain 
\begin{flalign*}
\pi(\bm{\beta} | \bm{z}, \bm{\tilde{y}}, \bm{y}) & \propto \mathcal{N}(\bm{\tilde{y}}; \bm{X} \bm{\beta},  \bm{I}_n) \mathcal{N}(\bm{\beta}; 0, \bm{D}^{-1}) \text{ for } \bm{D} \defeq \Diag(\bm{z} \tau_1^{-2} + (\bm{\mathrm{1}}_p-\bm{z}) \tau_0^{-2}) \\
& \propto \mathcal{N}(\bm{\beta}; \bm{\Sigma}^{-1} \bm{X}^\top \bm{\tilde{y}}, \bm{\Sigma}^{-1}) \text{ for } \bm{\Sigma} = \bm{X}^\top \bm{X} + \bm{D},  \\
\pi(\bm{z} | \bm{\beta}, \bm{\tilde{y}}, \bm{y}) &\propto \prod_{j=1}^p \Bernoulli \Big( z_i; \frac{q \mathcal{N}(\beta_j; 0, \tau_1^2)}{q \mathcal{N}(\beta_j; 0, \tau_1^2) + (1-q) \mathcal{N}(\beta_j; 0,  \tau_0^2)} \Big),  \text{ and }\\
\pi(\bm{\tilde{y}} | \bm{\beta}, \bm{z}, \bm{y}) &\propto \prod_{i=1}^n \mathcal{N} (\tilde{y}_i; \bm{x_i}^\top \bm{\beta}, 1) \mathrm{I} \big\{ \mathrm{I} \{\tilde{y}_i>0\} = y_i \big\}. 
\end{flalign*}
as required for probit regression in Algorithm \ref{algo:logistic_probit_spike_slab}.

\subsection{Logistic regression with spike-and-slab priors} 
\label{appendices:logistic}
We first describe the Bayesian logistic regression model considered.
Consider the logistic regression likelihood, where for each observation $i=1,\mydots,n$,
$\mathbb{P}( y_i = 1 | \bm{x_i}, \bm{\beta}) = 1 - \mathbb{P}( y_i = 0 | \bm{x_i}, \bm{\beta}) = \frac{\exp(\bm{x_i}^\top \bm{\beta})}{1+\exp(\bm{x_i}^\top \bm{\beta})}$ 
for $\bm{x_i}^\top$ the $i$-th row of the design matrix $\bm{X}$. 
We obtain $y_i = \mathrm{I} \{ \tilde{y}_i>0 \}$ where $\tilde{y}_i \overset{ind}{\sim} \mathrm{Logistic}(\bm{x_i}^\top \bm{\beta}, 1)$, corresponding to the logistic distribution centered about $\bm{x_i}^\top \bm{\beta}$ and scale parameter $1$. 

\subsubsection{Student's $t$-distribution based approximation of the logistic regression likelihood.}
Following \citet{obrien2004bayesianBIOMETRICS} and \citet{narisetty2019skinnyJASA}, we can approximate $\mathrm{Logistic}(\bm{x_i}^\top \bm{\beta}, 1)$ with $\bm{x_i}^\top \bm{\beta} + w t_{\nu}$, where $t_{\nu}$ denotes a $t$-distribution with $\nu$ degrees of freedom and $w$ is a multiplicative factor. The constants $\nu \defeq 7.3$ and $w^2 \defeq \frac{\pi^2(\nu-2)}{3 \nu}$ are chosen following \citet{obrien2004bayesianBIOMETRICS}, 
in order to match the variance of the logistic distribution and to minimize the integrated squared distance between the respective densities. The Gaussian scale representation of this $t$-distribution is 
\begin{equation}
    \tilde{y}_i | \bm{x_i}, \bm{\beta}, \tilde{\sigma}_i  \sim \mathcal{N}(\bm{x_i}^\top \bm{\beta}, \tilde{\sigma}_i^2),  \quad \tilde{\sigma}_i^2 \sim \InvGamma \Big(\frac{v}{2}, \frac{w^2 v}{2} \Big), \label{eq:logistic_approx}
\end{equation}
where each $\tilde{\sigma}_i^2$ is an augmented variable. The Bayesian logistic regression model is then given by
\begin{flalign}
z_j &\overset{i.i.d.}{\sim} \Bernoulli(q) \quad \text{ for all } j=1,\mydots,p \nonumber \\
\beta_j | z_j &\overset{ \ ind \ }{\sim} (1-z_j) \mathcal{N}(0, \tau_0^2) + z_j \mathcal{N}(0, \tau_1^2) 
\quad \text{ for all } j=1,\mydots,p \label{eq:logistic_spike_slab}  \\
\tilde{\sigma}^2_i &\overset{i.i.d.}{\sim} \InvGamma(\frac{\nu}{2}, \frac{w^2\nu}{2}) \nonumber \\
\tilde{y}_i | \bm{\beta}, \tilde{\sigma}^2_i &\overset{ \ ind \ }{\sim} \mathcal{N}(\bm{x_i}^\top \bm{\beta}, \tilde{\sigma}_i^2) 
\quad \text{ for all } i=1,\mydots,n \nonumber \\
y_i &\overset{ \ \ \quad \ \ }{=} \mathrm{I}\{ \tilde{y}_i > 0 \} \quad \text{ for all } i=1,\mydots,n.\nonumber
\end{flalign}
Let $\bm{\tilde{\sigma}}^2$ denote the vector with entries $\tilde{\sigma}_i^2$ for $i=1,\mydots,n$. For the  prior and likelihood in \eqref{eq:logistic_spike_slab}, the posterior density of $(\bm{\beta}, \bm{z}, \bm{\tilde{y}}, \bm{\tilde{\sigma}}^2)$ on $\mathbb{R}^p \times \{0,1\}^p \times \mathbb{R}^n \times (0,\infty)^n$ is given by
\begin{flalign}
\pi(\bm{\beta}, \bm{z}, \bm{\tilde{y}}, \bm{\tilde{\sigma}}^2 | \bm{y}) 
\propto 
& \prod_{j=1}^p \big( q \mathcal{N}(\beta_j; 0, \tau_1^2) \big)^{z_j} \big((1-q) \mathcal{N}(\beta_j; 0, \tau_0^2)\big)^{1-z_j}
\nonumber \\ 
& \prod_{i=1}^n \mathrm{I} \big\{ \mathrm{I} \{\tilde{y}_i>0\} = y_i \big\} \mathcal{N} (\tilde{y}_i; \bm{x_i}^\top \bm{\beta}, \tilde{\sigma}_i^2) \InvGamma \Big(\tilde{\sigma}_i^2 ; \frac{\nu}{2}, \frac{w^2 \nu}{2} \Big).
\label{eq:logistic_reg_posterior}
\end{flalign}
From \eqref{eq:logistic_reg_posterior}, we can calculate the conditional distributions. Let $\bm{W}=\Diag(\bm{\tilde{\sigma}}^2)$. We obtain 
\begin{flalign*}
\pi(\bm{\beta} | \bm{z}, \bm{\tilde{y}}, \bm{\tilde{\sigma}}^2, \bm{y}) & \propto \mathcal{N}(\bm{\tilde{y}}; \bm{X} \bm{\beta}, \bm{W}) \mathcal{N}(\bm{\beta}; 0, \bm{D}^{-1}) \text{ for } \bm{D} \defeq \Diag(\bm{z} \tau_1^{-2} + (\bm{\mathrm{1}}_p-\bm{z}) \tau_0^{-2}) \\
& \propto \mathcal{N}(\bm{\beta}; \bm{\Sigma}^{-1} \bm{X}^\top \bm{W}^{-1} \bm{\tilde{y}}, \bm{\Sigma}^{-1}) \text{ for } \bm{\Sigma} = \bm{X}^\top \bm{W}^{-1} \bm{X} + \bm{D}, \\
\pi(\bm{z} | \bm{\beta}, \bm{\tilde{y}}, \bm{\tilde{\sigma}}^2, \bm{y}) &\propto \prod_{j=1}^p \Bernoulli \Big( z_i; \frac{q \mathcal{N}(\beta_j; 0, \tau_1^2)}{q \mathcal{N}(\beta_j; 0,  \tau_1^2) + (1-q) \mathcal{N}(\beta_j; 0, \tau_0^2)} \Big), \\
\pi(\bm{\tilde{y}} | \bm{\beta}, \bm{z}, \bm{\tilde{\sigma}}^2, \bm{y}) &\propto \prod_{i=1}^n \mathcal{N} (\tilde{y}_i; \bm{x_i}^\top \bm{\beta}, \tilde{\sigma}_i^2) \mathrm{I} \big\{ \mathrm{I} \{\tilde{y}_i>0\} = y_i \big\}  , \text{ and } \\
\pi(\bm{\tilde{\sigma}}^2 | \bm{\beta}, \bm{z}, \bm{\tilde{y}}, \bm{y}) &\propto \prod_{i=1}^n \mathcal{N} (\tilde{y}_i; \bm{x_i}^\top \bm{\beta}, \tilde{\sigma}_i^2)
\InvGamma \Big(\tilde{\sigma}_i^2 ; \frac{\nu}{2}, \frac{w^2 \nu}{2} \Big) \\
&\propto \prod_{i=1}^n \InvGamma \Big( \tilde{\sigma}_i^2 ; \frac{\nu+1}{2}, \frac{w^2 \nu + (\tilde{y}_i - \bm{x_i}^\top \bm{\beta} )^2}{2} \Big).
\end{flalign*}
as required for logistic regression in Algorithm \ref{algo:logistic_probit_spike_slab}.

\paragraph{A scalable Gibbs sampler for logistic regression.}
The computational bottleneck of existing Gibbs samplers for logistic regression is linked to sampling from the full conditional of $\bm{\beta} \in \mathbb{R}^p$. This is given by 
\begin{equation} \label{eq:full_conditional_beta_logistic}
    \bm{\beta}_{t+1} | \bm{z}_t, \bm{\tilde{\sigma}}^2_t \sim 
    \mathcal{N} \big(\bm{\Sigma}^{-1}_t \bm{X}^\top \bm{W}_t^{-1} \bm{\tilde{y}}, \bm{\Sigma}^{-1}_t \big)\quad {\rm for}\quad  \bm{\Sigma}_t = \bm{X}^\top \bm{W}_t^{-1} \bm{X} + \bm{D}_t,
\end{equation}
where $t$ indexes the iteration of the Markov chain, 
$\bm{W}_t$ is the diagonal matrix with the vector
$\bm{\tilde{\sigma}}^2_t$ populating its diagonal elements, and $\bm{D}_t$ is the diagonal matrix with the vector $\bm{z}_t \tau^{-2}_1 + (\bm{\mathrm{1}}_p-\bm{z}_t) \tau^{-2}_0$ populating its diagonal elements. To sample from \eqref{eq:full_conditional_beta_logistic}, we can use
the $\Omega(n^2p)$ sampler of \citet{bhattacharya2016fastBIOMETRIKA}, which is given in Algorithm \ref{algo:fast_mvn_bhattacharya_logistic}. 

\begin{algorithm}[tb]
\caption{An $\Omega(n^2p)$ sampler of \eqref{eq:full_conditional_beta_logistic}
\citep{bhattacharya2016fastBIOMETRIKA}}
   \label{algo:fast_mvn_bhattacharya_logistic}
\begin{algorithmic}
  \STATE Sample $\bm{r} \sim \mathcal{N}(0, \bm{I}_p)$, $\bm{\xi} \sim \mathcal{N}(0,\bm{I}_n)$.
  \STATE Set $\bm{u} = \bm{D}_t^{-\frac{1}{2}}\bm{r}$ and calculate $\bm{v} = \bm{W}_t^{-1/2}\bm{X}\bm{u} + \bm{\xi}$.
  \STATE Set $\bm{v}^* = \bm{M}_t^{-1} ( \bm{W}_t^{-1/2} \bm{\tilde{y}} - \bm{v})$ for $\bm{M}_t = \bm{I}_n + \bm{W}_t^{-1/2} \bm{X} \bm{D}_t^{-1} \bm{X}^\top \bm{W}_t^{-1/2}$.
  \STATE {\bfseries Return} $\bm{\beta} = \bm{u} +  \bm{D}_t^{-1} \bm{X}^\top \bm{W}_t^{-1/2} \bm{v}^*$.
\end{algorithmic}
\end{algorithm}

\begin{subequations}
Following the strategy in Section \ref{subsection:S3}, \sss for logistic regression uses pre-computation to reduce the computational cost of Algorithm \ref{algo:fast_mvn_bhattacharya_logistic}. Using the notation from Section \ref{subsection:S3} with $\bm{M}_t \defeq \bm{I}_n + \bm{W}_t^{-1/2} \bm{X} \bm{D}_t^{-1} \bm{X}^\top \bm{W}_t^{-1/2}$, we note
\begin{flalign}
\bm{M}_t &= \bm{I}_n +
\bm{W}_t^{-1/2} \big( 
\bm{\tilde{M}}_{\tau_0} - I_n 
+ (\tau_1^2-\tau_0^2)\bm{X}_{A^c_t} \bm{X}_{A^c_t}^T
\big) \bm{W}_t^{-1/2} \label{eq:pre_compute_correction3a} \\
&= \bm{I}_n +
\bm{W}_t^{-1/2} \big( 
\bm{\tilde{M}}_{\tau_1} - I_n 
+ (\tau_0^2-\tau_1^2)\bm{X}_{A^c_t} \bm{X}_{A^c_t}^T
\big) \bm{W}_t^{-1/2} \label{eq:pre_compute_correction3b} \\
&= \bm{I}_n +
\bm{W}_t^{-1/2} \big(
\bm{W}_{t-1}^{1/2}(\bm{M}_{t-1} - I_n)\bm{W}_{t-1}^{1/2} + \bm{X}_{\Delta_t} \bm{C}_{\Delta_t} \bm{X}_{\Delta_t}^T
\big) \bm{W}_t^{-1/2} \label{eq:pre_compute_correction3c}.
\end{flalign}
In \eqref{eq:pre_compute_correction3a}~--~\eqref{eq:pre_compute_correction3c},
calculating the matrix products 
$\bm{X}_{A_t} \bm{X}_{A_t}^\top$, $\bm{X}_{A^c_t} \bm{X}_{A^c_t}^\top$, and $\bm{X}_{\Delta_t}  \bm{C}_{\Delta_t} \bm{X}_{\Delta_t}^\top$ requires 
$\mathcal{O}(n^2 \|\bm{z}_t\|_1)$, $\mathcal{O}(n^2 (p-\|\bm{z}_t\|_1))$, and
$\mathcal{O}(n^2 \delta_t)$ cost respectively. Given 
$\bm{\tilde{M}}_{\tau_0}$, $\bm{\tilde{M}}_{\tau_1}$, $\bm{M}_{t-1}$, and $\bm{z}_{t-1}$, 
we evaluate whichever matrix product in 
\eqref{eq:pre_compute_correction3a}~--~\eqref{eq:pre_compute_correction3c} has minimal computational cost and thereby calculate $\bm{M}_t$ at the reduced cost of 
$\mathcal{O}(n^2 p_t)$ where $p_t \defeq \min \{ \|\bm{z}_t\|_1, p-\|\bm{z}_t\|_1, \delta_t \}$.
\label{eq:pre_compute_correction3}
\end{subequations}
To calculate $\bm{M}_t^{-1}$, we calculate $\bm{M}_t^{-1}$ by directly inverting the calculated matrix $\bm{M}_t$, which requires $\mathcal{O}(n^3)$ cost. Overall, this strategy 
reduces the computational cost of calculating the matrices $\bm{M}_{t}$ and $\bm{M}_{t}^{-1}$ from
$\Omega(n^2p)$ to $\mathcal{O}(\max\{n^2 p_t,n^3\})$.

\paragraph{Extensions to Scalable Spike-and-Slab for logistic regression.}
Suppose the matrices $\bm{X}^\top \bm{X}$ is pre-computed. This initial step requires $\mathcal{O}(n p^2)$ computational cost and $\mathcal{O}(p^2)$ memory. Then the matrices $\bm{X}_{A_t}^\top \bm{X}_{A_t}$ and $\bm{X}_{A^c_t}^\top \bm{X}_{A^c_t}$ in \eqref{eq:pre_compute_correction3a}~--~\eqref{eq:pre_compute_correction3b} correspond to pre-computed sub-matrices of $\bm{X}^\top \bm{X}$, and
calculating $\bm{M}_t$ using \eqref{eq:pre_compute_correction3a}~--~\eqref{eq:pre_compute_correction3b}
each iteration $t$ involves matrix addition and diagonal matrix multiplication which only requires $\mathcal{O}(n^2)$ cost. To sample from \eqref{eq:full_conditional_beta_logistic}, we calculate $\bm{M}_t^{-1}$ by directly inverting the calculated matrix $\bm{M}_t$ from \eqref{eq:pre_compute_correction1}, which requires $\mathcal{O}(n^3)$ cost.
Overall, now the Gibbs samplers for logistic regression requires $\mathcal{O}(\max\{ n^3, np \})$ computational cost at iteration $t$, 
which is an improvement compared to \sss.

\section{Experiment Details} \label{appendices:experiments}
\paragraph{Figure \ref{fig:stat_comparison} of Section \ref{section:comparison}.}
In Figure \ref{fig:stat_comparison}, we use the same prior 
hyperparameters for all the algorithms. Following \citet{narisetty2019skinnyJASA},
we choose $\tau_0^2 = \frac{1}{n}$, $\tau_1^2 = \max \{ \frac{p^{2.1}}{100n},1 \}$ and  
$q = \mathbb{P}(z_j=1)$ such that 
$\mathbb{P}(\sum_{j=1}^p \mathrm{I}\{z_j=1\} > K)=0.1$ for 
$K=\max \{ 10, \log n \}$. 
The true positive rate (TPR) and the false discovery rate (FDR) correspond to the proportion of non-zero and zero components of $\beta^*_j$ that are correctly selected respectively. They are calculated as 
$\frac{1}{s} \sum_{j=1}^s \mathrm{I} \{ \mathbb{P}_{\pi}(z_{j}=1) > 0.5 \}$
and $ \frac{1}{p-s} \sum_{j=s+1}^p \mathrm{I} \{ \mathbb{P}_{\pi}(z_{j}=1) > 0.5 \}$ respectively, where the marginal posterior probabilities $\pi_j \defeq \mathbb{P}_{\pi}(z_{j}=1)$ are estimated by $\hat{\pi}_j \defeq \frac{1}{4000} \sum_{t=\mathrm{1001}}^{\mathrm{5000}} z_{j,t}$
for sample points $(\bm{z}_{t})_{t \geq 0}$ generated using \sss or Skinny Gibbs. 
The lines in Figure \ref{fig:stat_comparison} correspond to the average TPR and FDR across $20$ independently generated datasets, and the grey bands correspond to one standard error of the averages. 

\section{Dataset Details}
\label{appendices:datasets}
\paragraph{Synthetic continuous response dataset in Section \ref{subsection:complexity_analysis}.}
In Figure \ref{fig:comp_cost_sims}, synthetic linear regression datasets 
are considered. For number of observations $n$ and number of covariates $p$, 
we generate a design matrix $\bm{X} \in \mathbb{R}^{n \times p}$ such that 
each $[\bm{X}]_{i,j} \overset{i.i.d.}{\sim} \mathcal{N}(0, 1)$ for all 
$1 \leq i \leq n$ and $1 \leq j \leq p$, which is then scaled to 
ensure each column has a mean of $0$ and a standard error of $1$. 
We choose the true signal $\bm{\beta}^* \in \mathbb{R}^p$ such that 
$\bm{\beta}^*_j = 2 \mathrm{I}\{ j \leq s \}$, where $s$ is the 
sparsity parameter corresponding to the number of non-zero components.
Given $\bm{X}$ and $\bm{\beta}^*$, we generate $\bm{y} = \bm{X}\bm{\beta}^* + \sigma^* \bm{\epsilon}$ for 
$\bm{\epsilon} \sim \mathcal{N}(0,\bm{I}_n)$, where 
$\sigma^*=2$ is the Gaussian noise standard deviation. 

\paragraph{Synthetic binary response dataset in 
Section \ref{section:comparison}.} 
In Figures \ref{fig:time_comparison} and \ref{fig:stat_comparison}, 
synthetic binary classification datasets are considered. 
For number of observations $n$ and number of covariates $p$, 
we generate a design matrix $\bm{X} \in \mathbb{R}^{n \times p}$ such that 
each $[\bm{X}]_{i,j} \overset{i.i.d.}{\sim} \mathcal{N}(0, 1)$ for all 
$1 \leq i \leq n$ and $1 \leq j \leq p$, which is then scaled to 
ensure each column has a mean of $0$ and a standard error of $1$. 
We choose the true signal $\bm{\beta}^* \in \mathbb{R}^p$ such that 
$\bm{\beta}^*_j = 2^{\frac{9-j}{4}} \mathrm{I}\{ j \leq s \}$, 
where $s$ is the sparsity parameter corresponding to the number of non-zero components.
Given $\bm{X}$ and $\bm{\beta}^*$, we generate $y_i = \mathrm{I}\{ \tilde{y}_i > 0\}$ for
$\tilde{y}_i \sim \mathrm{Logistic}(\bm{x_i}^\top \bm{\beta}^*, 1)$ for $i=1\mydots,n$, 
where $\bm{x_i}^\top$ is the $i$-th row of $\bm{X}$ and 
$\mathrm{Logistic}(\bm{x_i}^\top \bm{\beta}^*, 1)$ is the Logistic distribution
with mean $\bm{x_i}^\top \bm{\beta}^*$ and scale parameter $1$. 

\paragraph{Datasets in Section \ref{section:applications}.} 

The Malware detection dataset from the UCI machine learning repository 
\citep{Dua2019UCI} has $n=373$ observations with binary responses and 
$p=503$ covariates, and is publicly available on 
\url{www.kaggle.com/piyushrumao/malware-executable-detection}.

The Borovecki, Chowdary, Chin and Gordon datasets are all high-dimensional
microarray datasets. They are publicly available 
on the $\mathrm{datamicroarray}$ package in $\mathrm{R}$. 
The Borovecki dataset has $n=31$ observations with binary responses and $p=22283$
covariates. 
The Chowdary dataset has $n=104$ observations with binary responses 
and $p=22283$ covariates. 
The Chin dataset has $n=118$ observations with binary 
responses and $p=22215$ covariates.
The Gordon dataset has $n=181$ observations with binary 
responses and $p=12533$ covariates.

The PCR GWAS dataset has $n=60$ observations with continuous responses
and $p=22575$ covariates, and is publicly available on 
\url{www.ncbi.nlm.nih.gov/geo} (accession number $GSE3330$). 
The Lymph Node GWAS dataset has $n=148$ observations with binary responses 
and $p=4514$ covariates, and has been previously considered
\citep{hans2007shotgunJASA,liang2013bayesianJASA,narisetty2019skinnyJASA}. 
The Maize GWAS dataset has $n=2266$ observations with continuous responses
and $p=98385$ covariates, and has been previously considered
\citep{romay2013genotypingGENOMEBIO, liu2016iterativePLOSGENETICS, zeng2017nonparametricNATURECOMM}.
The Lymph Node GWAS and the Maize GWAS datasets are not publicly available. 

The synthetic continuous dataset has $n=1000$ observations
and $p=50000$ covariates. The design matrix $\bm{X}$ is generates such that 
each $[\bm{X}]_{i,j} \overset{i.i.d.}{\sim} \mathcal{N}(0, 1)$ for all 
$1 \leq i \leq n$ and $1 \leq j \leq p$, which is then scaled to 
ensure each column has a mean of $0$ and a standard error of $1$. 
The true signal $\bm{\beta}^* \in \mathbb{R}^p$ is chosen such that 
$\beta^*_j = 2^{\frac{9-j}{4}}\mathrm{I}\{ j \leq s \}$, 
where $s$ is the sparsity parameter corresponding to the number of non-zero components.
Given $\bm{X}$ and $\bm{\beta}^*$, we generate $\bm{y} = \bm{X}\bm{\beta}^* + \sigma^* \bm{\epsilon}$ for 
$\bm{\epsilon} \sim \mathcal{N}(0,\bm{I}_n)$, where 
$\sigma^*=2$ is the Gaussian noise standard deviation. 
The synthetic binary classification dataset is generated as in 
Section \ref{section:comparison}, with $n=1000$ observations
and $p=50000$ covariates. 

\section{Additional Experiments}
\label{appendices:extra_experiments}

\paragraph{Variable selection performance as a function of time or number of iterations.}

\begin{figure}[!]
\vskip 0.2in
\begin{center}
\centerline{\includegraphics[width=1\columnwidth]{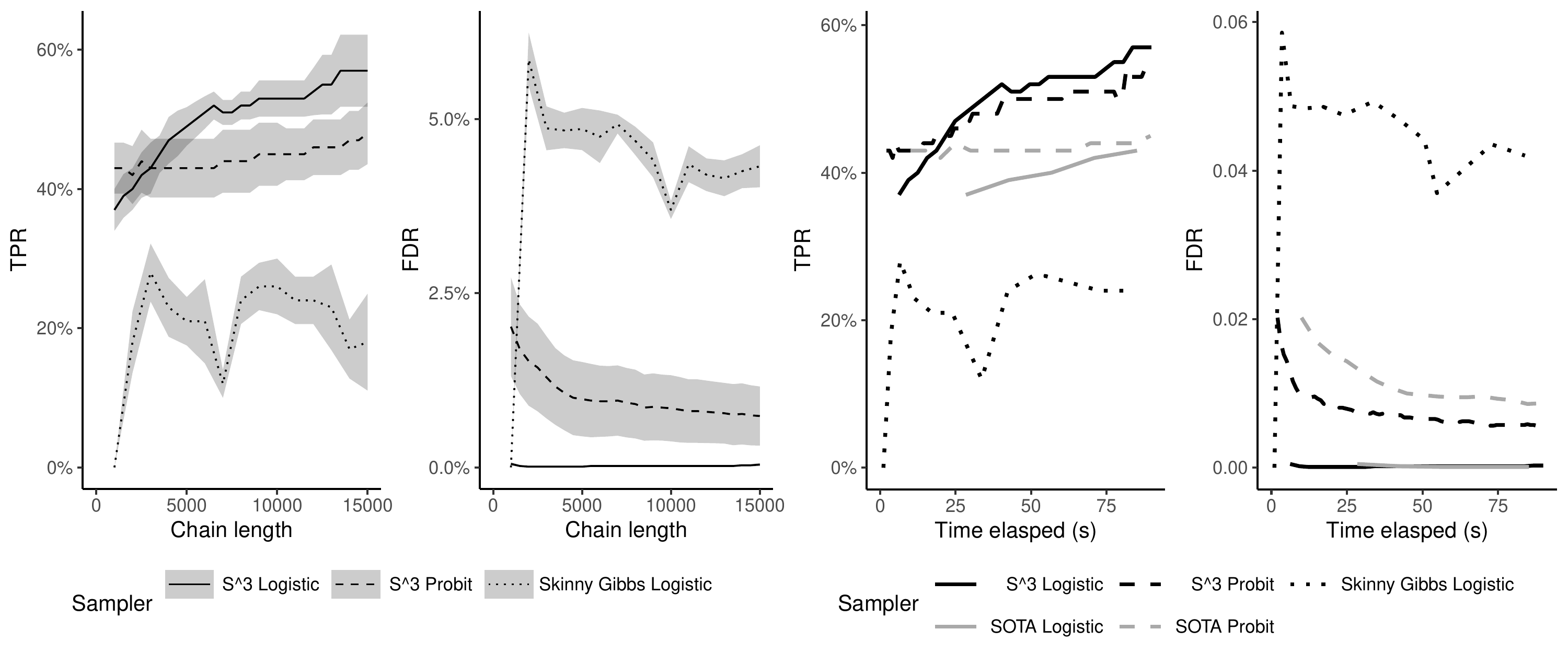}}
\vspace*{-0.05in}
\caption{
Avrage rue positive rate (TPR) and false discovery rate (FDR) plotted against 
the number of iterations and the total time elapsed in seconds. We consider 
 \sss, SOTA, and the Skinny Gibbs approximate sampler \citep{narisetty2019skinnyJASA}
applied to a synthetic binary classification dataset with $n=200$ observations and
$p=1000$ covariates. The TPR and FDR are averaged over 
$10$ independent chains, and one standard error bars are shown on the left and center-left plots and omitted on the right and center-right plots for visibility.
The SOTA sampler is omitted from the Left and Center-Left plots as its output has the same marginal 
distribution and statistical performance as \sss.
See Section \ref{appendices:extra_experiments} for details.
}
\label{fig:stat_comparison_over_time}
\end{center}
\vskip -0.2in
\end{figure}

Figure \ref{fig:stat_comparison_over_time} plots the average true positive rate (TPR) 
and the false discovery rate (FDR) of variable selection based on samples 
from \sss and Skinny Gibbs as the 
length of the chains are varied. The TPR and FDR are averaged over 
$10$ independent chains, and one standard error bars are shown. 
We consider a synthetic binary classification dataset
generated using a logistic regression model as in Section \ref{section:comparison}, 
with $n=200$ observations, $p=1000$ covariates, sparsity $s=10$, and an 
exponentially decaying sparse true signal $\bm{\beta}^* \in \mathbb{R}^p$ such that 
$\beta^*_j = 2^{\frac{9-j}{4}}$ for $j \leq s$ and $\beta^*_j=0$ for $j>s$.
The TPR and FDR are calculated as in Section \ref{section:comparison} with  
the marginal posterior probabilities $\pi_j \defeq \mathbb{P}_{\pi}(z_{j}=1)$  
now estimated by $\hat{\pi}_j \defeq \frac{1}{T-999} \sum_{t=\mathrm{1000}}^{\mathrm{T}} z_{j,t}$
for a burn-in of $1000$, a varying chain length $T \geq 1000$, and sample points $(\bm{z}_{t})_{t \geq 0}$ generated using \sss or Skinny Gibbs. We use the same prior 
hyperparameters for all the algorithms, which are chosen 
according to \citet{narisetty2019skinnyJASA}.

Figure \ref{fig:stat_comparison_over_time} Left and Center-Left plot the 
TPR and FDR against the chain length $T$. It shows that \sss for both logistic
and probit regression have higher TPR and lower FDR than Skinny Gibbs for all chain 
lengths. Furthermore, \sss for logistic regression has higher TPR and lower FDR than \sss for probit regression, 
which is expected as the synthetic dataset for this example is generated using a 
logistic regression model. The SOTA sampler is omitted from the Left and Center-Left plots as its output has the same marginal 
distribution and statistical performance as \sss.
Figure \ref{fig:stat_comparison_over_time} Center-Right and Right
plot the TPR and FDR against total time elapsed in seconds to generate samples 
using \sss, SOTA, or Skinny Gibbs chains with a burn-in of $1000$ iterations. The standard error bars are now omitted for better visibility. 
For each time budget, we observe better variable selection performance from \sss when compared with the slower SOTA implementation or with Skinny Gibbs.

\paragraph{Effective Sample Size of \sss for the datasets in Section \ref{section:applications}.}
Figure \ref{fig:ESS_plot} shows the Effective Sample Size per iteration 
and per unit of time (in seconds) of \sss and the SOTA sampler for the datasets in Section \ref{section:applications}.
The ESS is calculated using the $\mathrm{mcmcse}$ package \cite{Flegal2021mcmcse, vats2020multivariateBIOMETRIKA}
for one \sss chain of length $10000$ iterations with a burn-in of $1000$ iterations for each dataset.
The average ESS of the $\beta$ components are then plotted.
Figure \ref{fig:ESS_plot} Right shows that \sss has significantly higher ESS per second compared
to the corresponding SOTA sampler for all the datasets considered.

\begin{figure}[!]
\vskip 0.2in
\begin{center}
\centerline{\includegraphics[width=0.9\columnwidth]{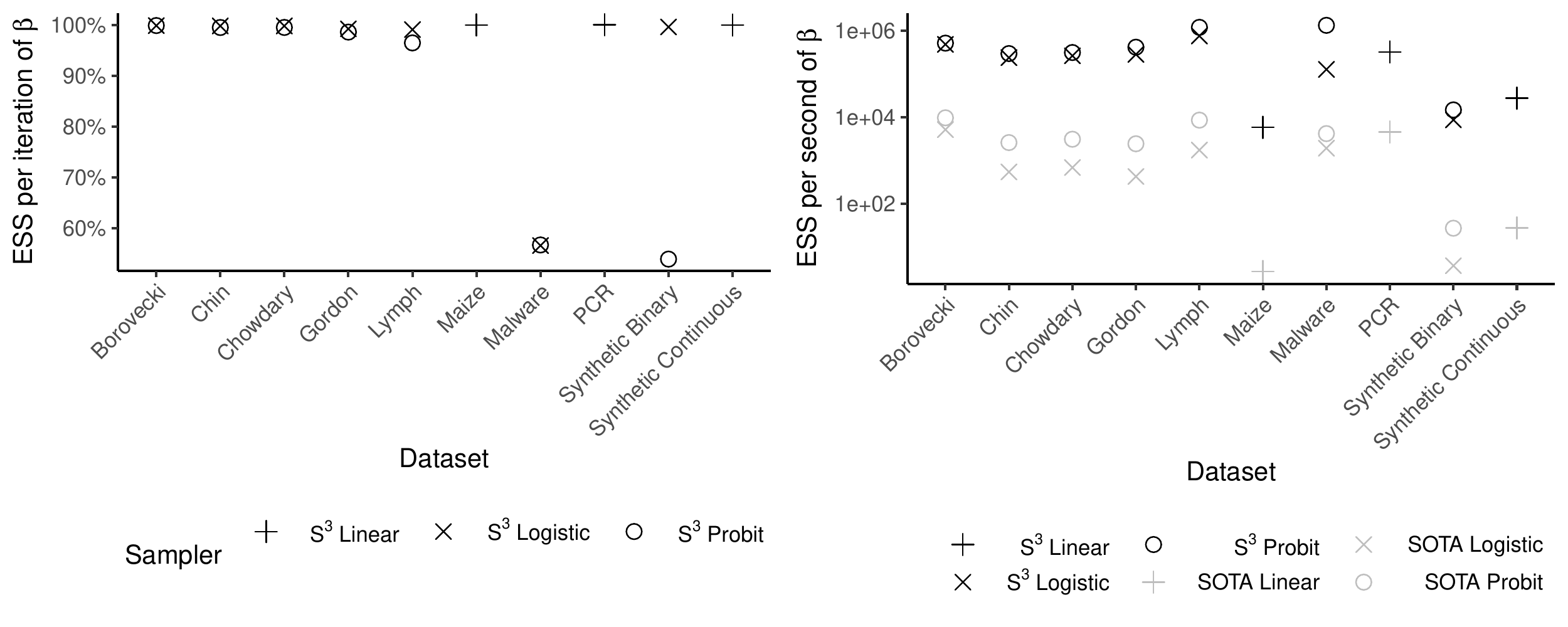}}
\vspace*{-0.05in}
\caption{
Effective sample size (ESS) per iteration and per second of \sss and the SOTA sampler for some of the datasets in Section \ref{section:applications}. The 
ESS is calculated using one \sss chain of length $10000$ iterations with a burn-in of $1000$ iterations.
The ESS per iteration of the SOTA sampler is omitted from the Left plot as it implements the same Gibbs sampler as \sss.
}
\label{fig:ESS_plot}
\end{center}
\vskip -0.2in
\end{figure}

\paragraph{Performance metrics for the datasets in Section \ref{section:applications}.}
Figures \ref{fig:Borovecki_plot}~--~\ref{fig:PCR_plot} show various performance metrics of 
\sss for some of the datasets considered in Section \ref{section:applications}. 
Figures \ref{fig:Borovecki_plot}~--~\ref{fig:PCR_plot} (Left) plot the marginal posterior 
probability estimates $\hat\pi_j$ against $j$ in the decreasing order of $\hat\pi_j$s, 
following the setup in Figure \ref{fig:applications1}. 
For datasets with continuous valued responses, $\hat\pi_j$s are based on 
samples from \sss for linear regression. 
For datasets with binary valued responses, $\hat\pi_j$s are based on 
samples from \sss for logistic and probit regression, and the Skinny Gibbs 
sampler from logistic regression.
We use samples from $5$ independent 
chains of length $10000$ iterations with a burn-in of $1000$ iterations.
Estimates based on samples from the SOTA sampler are not shown, 
as they implement the same Gibbs sampler as \sss
(other than possible numerical discrepancies, as discussed in Section
\ref{section:applications}). 

Figures \ref{fig:Borovecki_plot}~--~\ref{fig:PCR_plot} (Center) show the 
average time taken per iteration with one standard error bars for \sss, the SOTA sampler,  
and the Skinny Gibbs sampler based on $5$ independent chains of length $10000$ iterations.

Figures \ref{fig:Borovecki_plot}~--~\ref{fig:PCR_plot} (Right) show the 10-fold
cross-validation average root-mean-square error (RMSE) against the total time elapsed to run 
one \sss and one SOTA chain.
To compute this evaluation, we partition the observed dataset into $10$ folds uniformly at random and, for each fold $k$, run a chain conditioned on all data outside of fold $k$ and evaluate its performance on the held-out data in the $k$-th fold.
The average RMSE is calculated as 
$ \frac{1}{10} \sum_{k=1}^{10} r_k$, where $r_k$ is the RMSE for the
$k^{th}$ fold. For datasets with continuous valued responses, the 
quantities $r_k$ for linear regression are calculated as 
$(\frac{1}{|D_k|} \sum_{i \in D_k} (y_i - \hat{y}_i)^2)^{1/2}$ where 
$D_k$ is the $k^{th}$ fold, $\hat{y}_i \defeq \frac{1}{T-1000} \sum_{t=1001}^{T} x_i^T \beta_{t}$
are the predicted responses, and $(\beta_{t})_{t \geq 0}$ are samples from \sss and SOTA
targeting the posterior distribution of the $k^{th}$ training set.
For datasets with binary valued responses, the quantities $r_k$ 
are calculated as $(\frac{1}{|D_k|} \sum_{i \in D_k} (y_i - \hat{p}_i)^2)^{1/2}$, 
where $D_k$ is the $k^{th}$ fold, 
$\hat{p}_i \defeq \frac{1}{T-1000} \sum_{t=1001}^{T} \mathrm{Logistic}(x_i^T \beta_{t})$
and $\hat{p}_i \defeq \frac{1}{T-1000} \sum_{t=1001}^{T} \Phi(x_i^T \beta_{t})$ are the 
predicted probabilities for logistic and probit regression respectively, and $(\beta_{t})_{t \geq 0}$ are samples 
from \sss and SOTA targeting the posterior distribution of the $k^{th}$ training set.
Figures \ref{fig:Borovecki_plot}~--~\ref{fig:PCR_plot} (Right) plot the average RMSE 
against total time elapsed in seconds to generate samples using \sss or SOTA chains 
with a burn-in of $1000$ iterations.
The RMSE of the Skinny Gibbs sampler is not available, 
as the $\mathrm{skinnybasad}$ package does not output the full chain trajectories required
for RMSE calculations.

\begin{figure}[!]
\vskip 0.2in
\begin{center}
\centerline{\includegraphics[width=1\columnwidth]{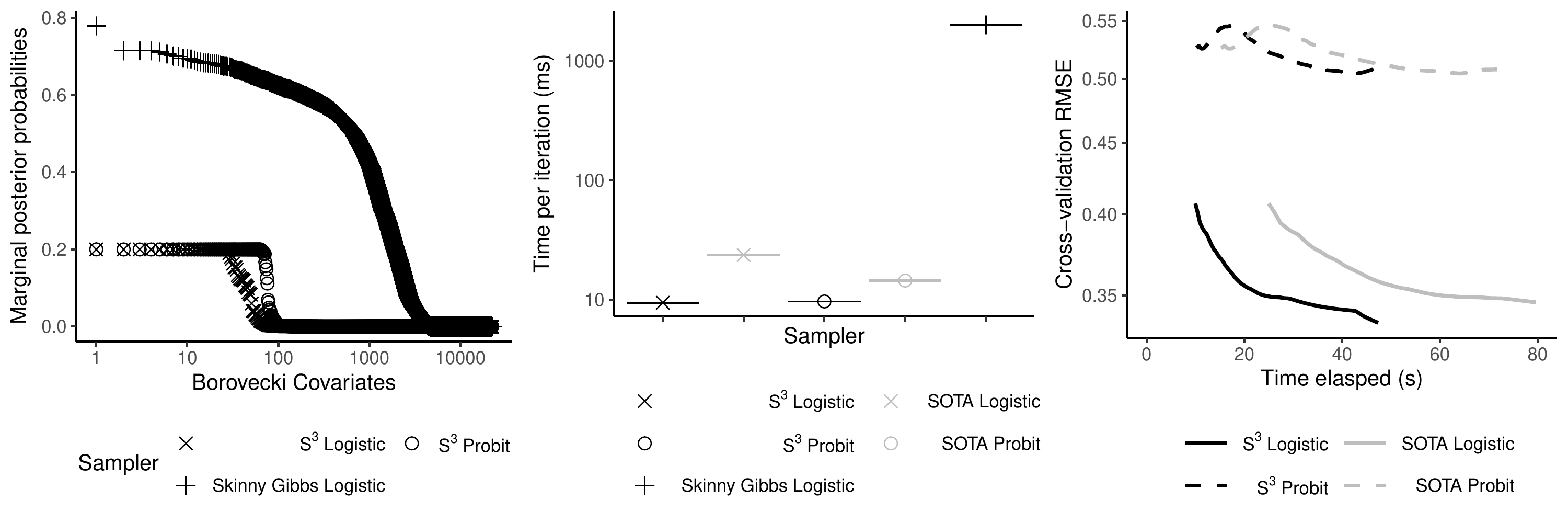}}
\vspace*{-0.05in}
\caption{
Borovecki dataset with $n=31$ observations, $p=22283$ covariates and binary valued responses.
}
\label{fig:Borovecki_plot}
\end{center}
\vskip -0.2in
\end{figure}

\begin{figure}[!]
\vskip 0.2in
\begin{center}
\centerline{\includegraphics[width=1\columnwidth]{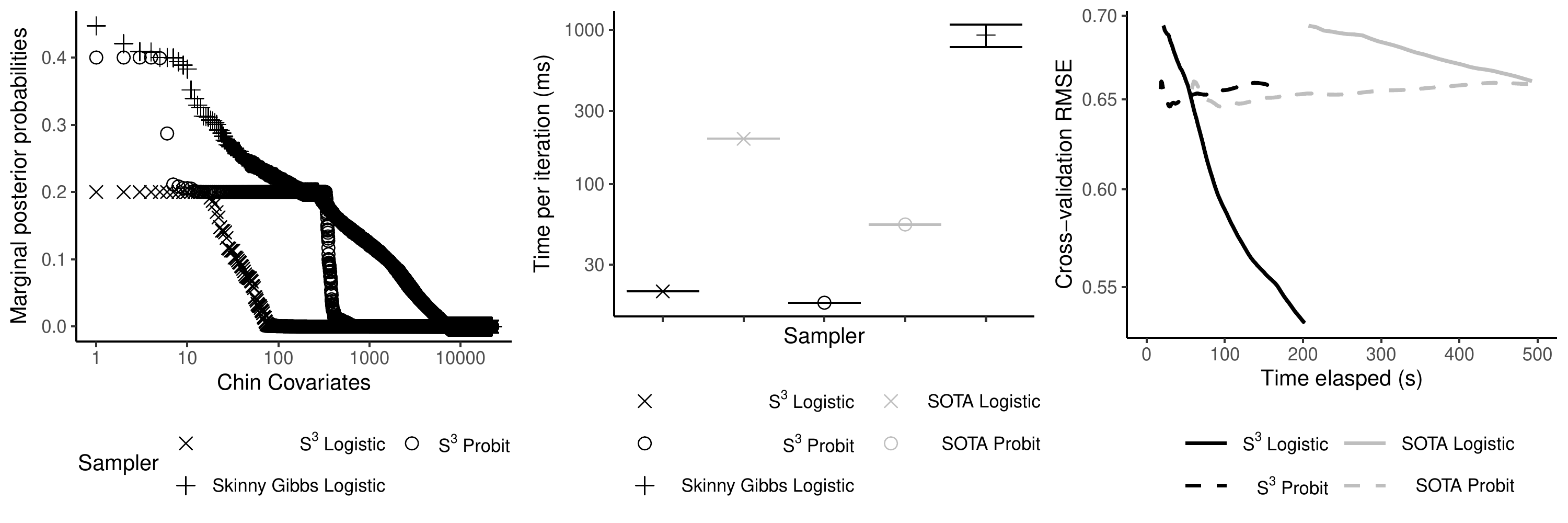}}
\vspace*{-0.05in}
\caption{
Chin dataset with $n=118$ observations, $p=22215$ covariates and binary valued responses.
}
\label{fig:Chin_plot}
\end{center}
\vskip -0.2in
\end{figure}

\begin{figure}[!]
\vskip 0.2in
\begin{center}
\centerline{\includegraphics[width=1\columnwidth]{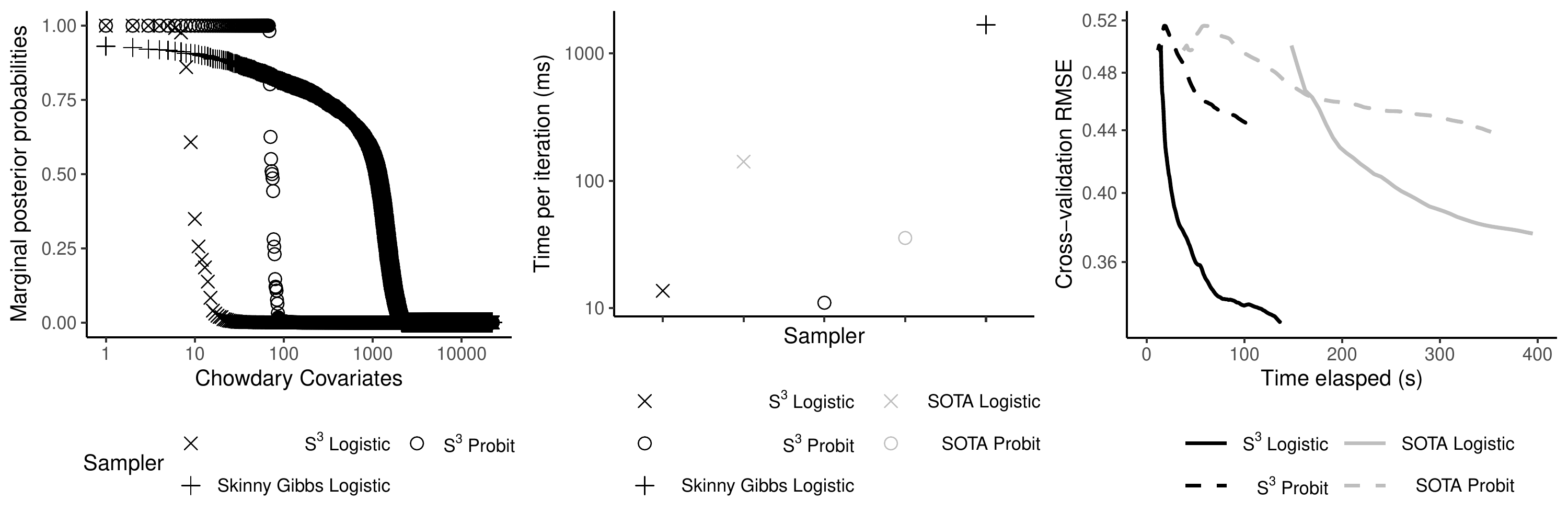}}
\vspace*{-0.05in}
\caption{
Chowdary dataset with $n=104$ observations, $p=22283$ covariates and binary valued responses.
}
\label{fig:Chowdary_plot}
\end{center}
\vskip -0.2in
\end{figure}

\begin{figure}[!]
\vskip 0.2in
\begin{center}
\centerline{\includegraphics[width=1\columnwidth]{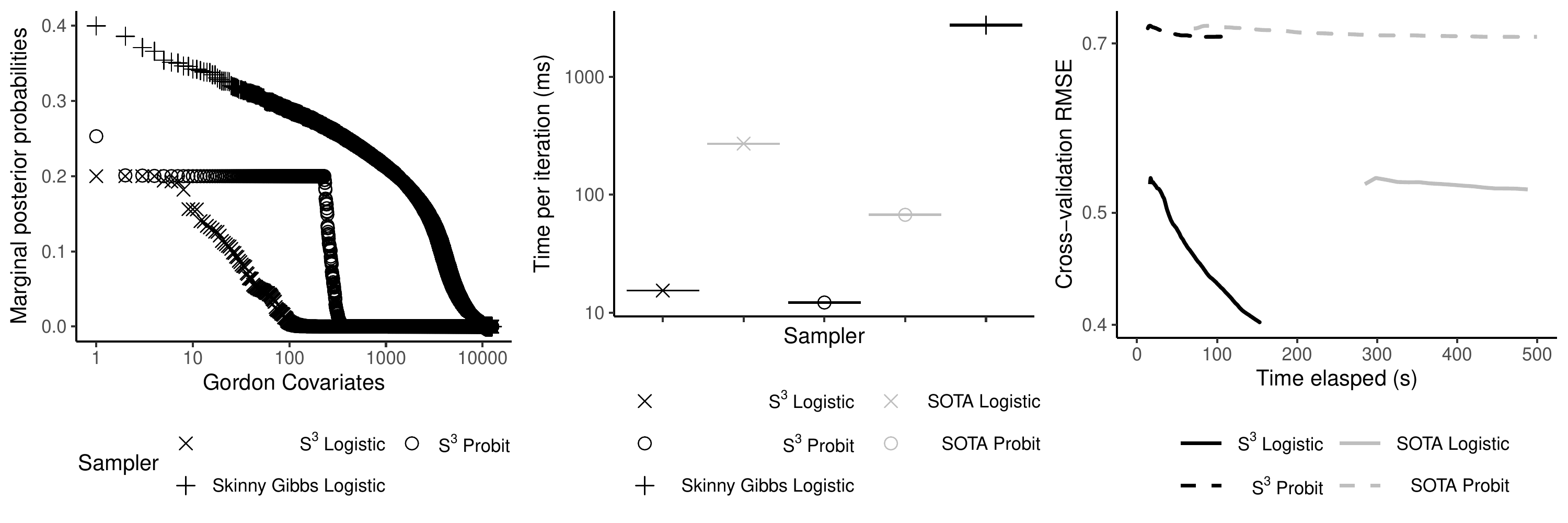}}
\vspace*{-0.05in}
\caption{
Gordon dataset with $n=181$ observations, $p=12533$ covariates and binary valued responses.
}
\label{fig:Gordon_plot}
\end{center}
\vskip -0.2in
\end{figure}

\begin{figure}[!]
\vskip 0.2in
\begin{center}
\centerline{\includegraphics[width=1\columnwidth]{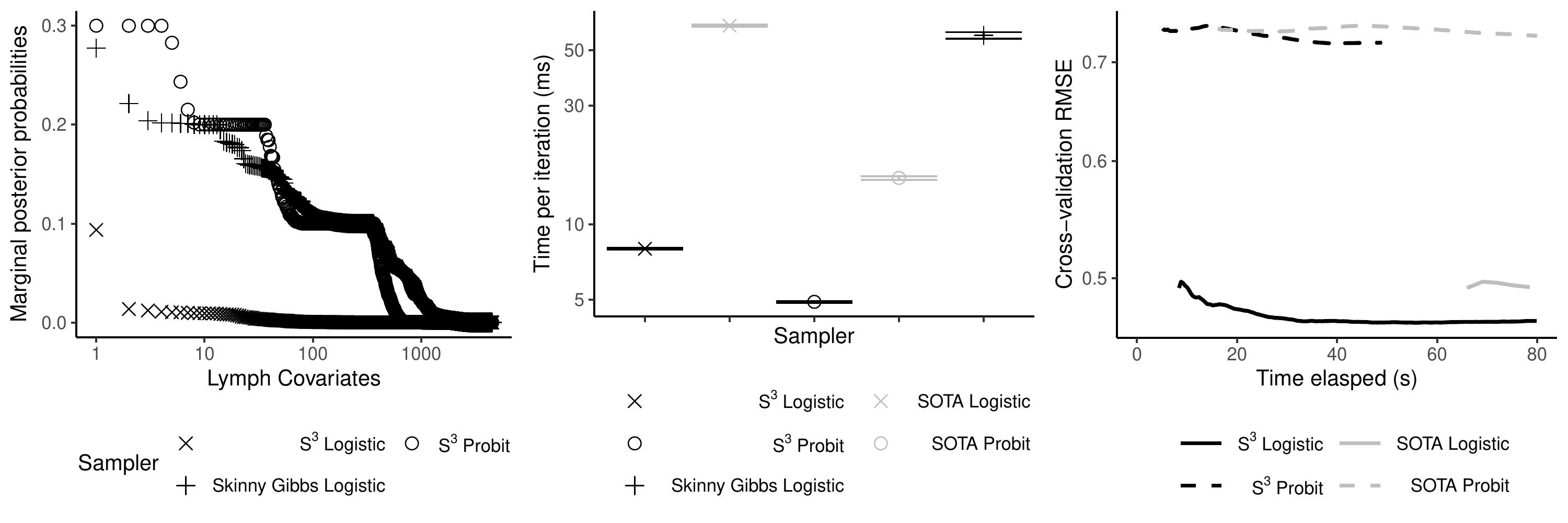}}
\vspace*{-0.05in}
\caption{
Lymph dataset with $n=148$ observations, $p=4514$ covariates and binary valued responses.
}
\label{fig:Lymph_plot}
\end{center}
\vskip -0.2in
\end{figure}

\begin{figure}[!]
\vskip 0.2in
\begin{center}
\centerline{\includegraphics[width=1\columnwidth]{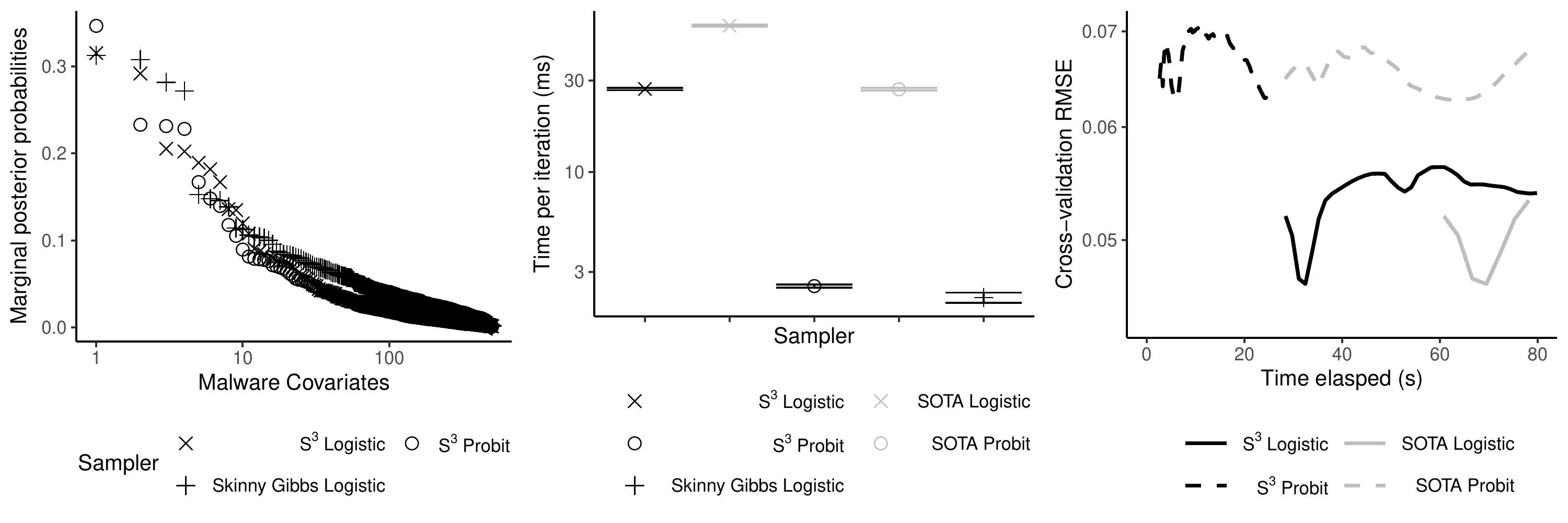}}
\vspace*{-0.05in}
\caption{
Malware dataset with $n=373$ observations, $p=503$ covariates and binary valued responses.
}
\label{fig:Malware_plot}
\end{center}
\vskip -0.2in
\end{figure}

\begin{figure}[!]
\vskip 0.2in
\begin{center}
\centerline{\includegraphics[width=1\columnwidth]{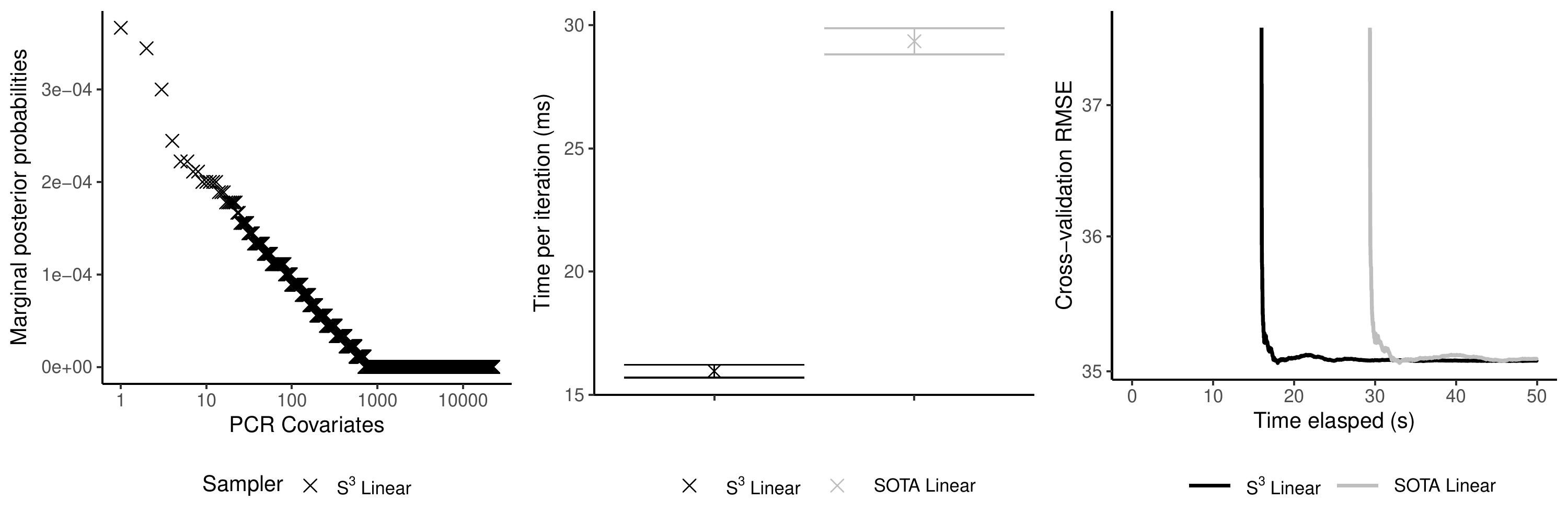}}
\vspace*{-0.05in}
\caption{
PCR dataset with $n=60$ observations, $p=22575$ covariates and continuous valued responses.
}
\label{fig:PCR_plot}
\end{center}
\vskip -0.2in
\end{figure}

\end{document}